\documentclass[12pt]{article}
\usepackage{amsfonts}
\usepackage[utf8]{inputenc}
\usepackage{xcolor}
\usepackage{amsmath}
\usepackage{hyperref}
\usepackage{amsthm}
\usepackage{tocloft}
\usepackage{authblk}
\usepackage{tikz-cd}

\theoremstyle{plain}
\usepackage[sorting=none,maxbibnames=99, backend=bibtex]{biblatex}
\bibliography{Bibliography}

\newtheorem{theorem}{Theorem}
\usepackage[margin=1.5in]{geometry}
\theoremstyle{definition}
\newtheorem{defn}[theorem]{Definition} 
\newtheorem{exmp}[theorem]{Example}
\newtheorem{lem}[theorem]{Lemma}
\newtheorem{cor}[theorem]{Corollary}
\newtheorem{rmk}[theorem]{Remark}
\newtheorem{prop}[theorem]{Proposition}

\newcommand{\Ext}{\mathrm{Ext}}
\newcommand{\Hom}{\mathrm{Hom}}
\newcommand{\im}{\mathop{\mathrm{im}}}
\newcommand{\coker}{\mathop{\mathrm{coker}}}
\newcommand{\ann}{\mathop{\mathrm{Ann}}}
\newcommand{\Tor}{{\mathrm{Tor}}}
\newcommand{\ZZ}{\mathbb{Z}}

\newcommand{\id}{\mathrm{id}}
\newcommand{\Cech}{\check{\mathrm{C}}}
\newcommand{\CH}{\check{\mathrm{H}}}
\renewcommand{\H}{\mathrm{H}}
\newcommand{\E}{\mathrm{E}}
\newcommand{\Res}{\mathrm{Res}}
\newcommand{\e}{\mathrm{e}}
\renewcommand{\i}{\mathrm{i}}

\title{Homological invariants of\\Pauli stabilizer codes}
\author[1]{Blazej Ruba}
\author[2]{Bowen Yang}
\affil[1]{\textit{Institute of Theoretical Physics, Jagiellonian University, \protect\\
prof. Łojasiewicza 11, 30-348 Kraków, Poland, \protect\\
email: blazej.ruba@gmail.com}}
\affil[2]{\textit{California Institute of Technology, Pasadena, CA 91125, \protect\\
email: byyang@caltech.edu}}

\date{\today}

\begin{document}
\maketitle
\begin{abstract}
We study translationally invariant Pauli stabilizer codes with qudits of arbitrary, not necessarily uniform, dimensions. Using homological methods, we define a series of invariants called charge modules. We describe their properties and physical meaning. The most complete results are obtained for codes whose charge modules have Krull dimension zero. This condition is interpreted as mobility of excitations. We show that it is always satisfied for translation invariant 2D codes with unique ground state in infinite volume, which was previously known only in the case of uniform, prime qudit dimension. For codes all of whose excitations are mobile we construct a~$p$-dimensional excitation and a~$(D-p-1)$-form symmetry for every element of the $p$-th charge module. Moreover, we define a braiding pairing between charge modules in complementary degrees. We discuss examples which illustrate how charge modules and braiding can be computed in practice. 
\end{abstract}

\newpage

\tableofcontents

\newpage

\section{Introduction}

Pauli stabilizer codes are spin systems whose ground state (and excitations) are described by eigenequations for a set of mutually commuting operators, each of which is a tensor product of finitely many Pauli matrices, or generalizations thereof called clock and shift matrices. Initially these models were studied as a~class of quantum error-correcting codes \cite{gottesman1996class, calderbank1997quantum}. Due to their mathematical tractability and nontrivial properties, they have become popular also as exactly solvable models of exotic phases of quantum matter. Qubits (or qudits) are typically placed on sites of a $D$-dimensional square lattice. Perhaps the most famous example is the toric code \cite{Fault}.

One may ask which quantum phases can be realized as Pauli stabilizer codes. It has been shown \cite{bombin2014structure,haah2021classification} that for codes on $\ZZ^2$ lattice with prime-dimensional qudits, stacks of toric codes are the only nontrivial phases with a~unique ground state in infinite volume. The story is richer for qudits of composite dimension. Namely, it was shown \cite{Z_4} that every abelian anyon model which admits a~gapped boundary \cite{kapustin2011topological} may be represented by a~Pauli stabilizer code\footnote{During final revisions of this manuscript we have learned about \cite{ellison2023pauli}, where it is claimed that abelian anyon models without gapped boundaries can also be realized if one uses Pauli subsystem codes.}. It was conjectured that the list of models constructed therein is exhaustive (up to finite depth quantum circuits and stabilization). There was even work on algorithmic determination of the corresponding abelian anyon model given a stabilizer code \cite{liang2023extracting}. The proposed classification depends on several assumptions, one of which is that all local excitations in Pauli stabilizer codes are mobile and hence can be created at endpoints of string operators. In~this paper we prove this, extending earlier results for prime-dimensional qudits. Stabilizer codes are even richer for $D>2$ \cite{haah2011local} due to the existence of so-called fractons: local excitations which can not be moved in any direction by acting with local operators. All this shows that mathematical study of stabilizer codes is an interesting and nontrivial problem. It~is also closely related to classification of Clifford Quantum Cellular Automata \cite{schlingemann2008structure,haah2021clifford,shirley2022three}.











Let us recall how similar classification problems were handled in other areas, e.g.\ algebraic topology. Historically, researchers first discovered some basic invariants, such as Euler characteristic or fundamental group. Later they developed more systematic methods, e.g. axiomatic ~(co)homology and homotopy theory. In our situation, the module of topological point excitations \cite{haah2013commuting} and (for~the case $D=2$) topological spin and braiding \cite{haah2021classification} are the known invariants. It is natural to look for machinery that produces their generalizations. Hopefully it will allow researchers to make progress in this classification problem.

In this article, we develop such tools for translationally invariant Pauli stabilizer codes with qudits of arbitrary (perhap not even uniform) dimension placed on a~lattice described by a finitely generated abelian group $\Lambda$. This setup incorporates infinitely extended as well as finite spatial directions. Of course physics crucially depends on $D = \mathrm{rk}(\Lambda)$ (the number of independent infinite directions). We describe stabilizer codes by symplectic modules over a~group ring $R$ of $\Lambda$ and their Lagrangian (or more generally, isotropic) submodules. This is closely related to the approach developed in \cite{haah2013commuting}. In contrast to treatment therein, the emphasis is on modules with direct physical interpretation, rather than their presentations with maps from free modules\footnote{The latter approach is very useful in concrete computations. We prefer ours in general considerations.}. 

We propose a definition of modules $Q^p$ of charges of $p$-dimensional excitations (anyons, fractons, strings etc.) for every non-negative integer $p$. The construction of $Q^p$ uses standard homological invariants of modules. In the case of local excitations ($p=0$), our definition agrees with the known one. For general $p$, the physical interpretation of mathematically defined $Q^p$ is most justified under the assumption that all charge modules have zero Krull dimension (which we interpret as the requirement that the excitations are mobile). In this case, we define for every element of $Q^p$ an operator with $(p+1)$-dimensional support which creates an excitation on its boundary. This excitation is uniquely defined modulo excitations which can be created by $p$-dimensional operators. Its mobility (moving around with $p$-dimensional operators) is established. We show also that every element of $Q^p$ gives rise to a~$(D-p-1)$-form symmetry \cite{gaiotto2015generalized}. Furthermore, a~braiding pairing between $Q^p$ and $Q^q$ (with $p+q=D-2$) is defined and its basic properties (such as symmetry) are established. 

It is natural to expect that for codes with only mobile excitations, the underlying abelian groups of $Q^p$ and pairings between them described above are (a part of) data of some Topological Quantum Field Theory (TQFT), e.g.\ an abelian higher gauge theory\footnote{Say, with action $\frac{1}{4 \pi} \sum_{p=1}^{D-1} \sum_{i,j} K_p^{ij} \int A_p^i \mathrm{d} A_{D-p}^j$, where $A_p^i$ are $p$-form $\mathrm{U}(1)$ gauge fields and $K$ matrices are non-degenerate and satisfy $K_p^{ij} = (-1)^{p+1} K^{ji}_{D-p}$.}. Such correspondence exists in every example known to authors. Modules $Q^p$ have more structure, which does not seem to be captured by a TQFT: they are acted upon by the group of translations. In some cases this allows to distinguish models with the same topological order which are distinct as Symmetry Enriched Topological (SET) phases with translational symmetry.

Section 2 details the mathematical set-up of translationally invariant stabilizer codes in terms of commutative algebra. Rudiments of symplectic geometry over group rings of $\Lambda$ are laid out here. Section 3 makes the connection between topological excitations and the functor $\Ext$. Section 4 discusses operations on stabilizer codes, e.g.\ coarse-graining and stacking. In particular we prove that charge modules are invariant to coarse-graining and that they provide obstructions to obtaining a system from a lower dimensional one by stacking. Section 5 ventures a definition of mobility for excitations in any dimension. We also include a proof for the conjecture that in any 2D code with unique ground state, all excitations are mobile and can be created with string operators. In Section 6, we specialize to codes with only mobile excitations. It is shown that in this case charges may be described by cohomology classes of a certain \v Cech complex. We show how to obtain interesting operators and physical excitations from \v Cech cocycles. Moreover, we define braiding in terms of a cup product in the \v Cech complex and show that our proposal reduces to what is expected for $D=2$. Several examples are worked out in Section 7. Some known mathematical definitions and facts used in the main text are reviewed in appendices: Gorenstein rings in Appendix A, local cohomology in Appendix B and \v Cech cohomology in Appendix C.

Let us mention some problems which are left unsolved in this work. Firstly, results of Section 6 are restricted to so-called Lagrangian stabilizer codes such that charge modules have Krull dimension zero. We would like to remove some of these assumptions in the future, for example to treat models with spontaneous symmetry breaking or fractons. Secondly, we did not prove that braiding is non-degenerate. We expect that this can be done by relating braiding to Grothendieck's local duality, in which we were so far unsuccessful. We expect also that the middle-dimensional braiding admits a distinguished quadratic refinement for $D=4k+2$ (which is already known to be true for $D=2$ from previous treatments) and that it is alternating (rather than merely skew-symmetric) for $D=4k$. Thirdly, it is not known in general to what extent invariants we defined determine a~stabilizer code, presumably up to symplectic transformations (corresponding to Clifford Quantum Cellular Automata), coarse graining and stabilization. We hope that in the future a one-to-one correspondence between equivalence classes of stabilizer codes with only mobile excitations and some (abelian) TQFTs will be established.

\section{Stabilizer codes and symplectic modules}

In order to obtain homological invariants of a stabilizer code, we need to translate it to the language of modules. In this section we generalize \cite{haah2013commuting} to codes with arbitrary (prime or composite) qudit dimensions. Multiple qudits are placed on each lattice site. A~$d$-dimensional qubit is acted upon by shift and clock matrices $X,Z$, which satisfy
\begin{equation}
    XZ = \e^{\frac{2 \pi \i}{d} } ZX, \qquad X^d = Z^d = 1.
\end{equation}
For brevity, products of $Z$ and $X$ (possibly acting on finitely many different qudits) and phase factors will be called Pauli operators. Unlike \cite{haah2013commuting}, our framework does not require qudits in a model to have a uniform dimension. Instead, an array of qudits with various dimensions populates each lattice site. We let $n$ be a common multiple of dimensions of all qudits in a model.  

All rings are commutative with unity and $\ZZ_n$ is the ring $\ZZ / n \ZZ$.

\begin{defn} \label{def:grp_rings}
Let $n$ be a positive integer and $\Lambda$ a finitely generated abelian group. $\ZZ_n[ \Lambda]$ is the group ring of $\Lambda$ over $\ZZ_n$. When $n$ and $\Lambda$ are clear from the context, we denote $R = \ZZ_n [\Lambda]$. For $\lambda \in \Lambda$, we denote the corresponding element of $R$ by $x^\lambda$. If $r = \sum_{\lambda \in \Lambda} r_\lambda x^{\lambda}$ (with all but finitely many $r_\lambda \in \ZZ_n$ equal to zero), we~call $r_0$ the scalar part of $r$. Moreover, we let $\overline r = \sum_{\lambda \in \Lambda} r_\lambda x^{- \lambda}$. Operation $r \mapsto \overline r$ is called the antipode.
\end{defn}

\begin{exmp}
Suppose that $\Lambda = \ZZ^D$. Then $R$ is the ring of Laurent polynomials in $D$ variables $x_1, \dots, x_D$, corresponding to $D$ elements of a basis of $\ZZ^D$. A~general element of $R$ is a sum of finitely many monomials $x_1^{\lambda_1} \cdots x_D^{\lambda_D}$ with $\ZZ_n$ coefficients; exponents $\lambda_i$ are in $\ZZ$. Here we use the more economical notation in which such monomial is simply denoted $x^\lambda$. One may think of $\lambda$ as a multi-index.
\end{exmp}

For a lattice $\Lambda$ with the same array of qudits on each site, ring $R=\ZZ_n[\Lambda]$ describes certain basic operations on Pauli operators. An element $x^\lambda$ translates a Pauli operator on the lattice by $\lambda \in \Lambda$, while a scalar $m\in \ZZ_n$ raises a Pauli operator to $m$-th power. As $n$ is a common multiple of qudit dimensions, taking the $n$-th power of any Pauli operator gives a scalar. This action endows the collection $P$ of all local Pauli operators modulo overall phases with an $R$-module structure. We will sometimes call elements of $P$ operators for conciseness.  Addition in $P$ corresponds to composition of operators, which is commutative because we are disregarding phases. Specifically, if qudits on each site have respective dimensions $n_1, \dots, n_q$, then $P$ is isomorphic to the module $\bigoplus_{j=1}^q \ZZ_{n_j}[\Lambda]^{\oplus 2}$. It is not a free module unless $n=n_1=n_2=\dots=n_q$. We will see that it nevertheless shares some homological properties of free modules, which is important in the study of invariants. In most cases, understanding of proofs is not necessary to read the remainder of the paper.

We will also define an antipode-sesquilinear symplectic form $\omega: P \times P \rightarrow R$ on $P$, which captures commutation relations satisfied by Pauli operators. More precisely, if $T,T'$ are Pauli operators corresponding to elements $p,p' \in P$, then
\begin{equation}
    TT' = \exp \left( \frac{2 \pi \i}{n} \omega(p,p')_0 \right) T'T.
\end{equation}
Thus it is the scalar part of $\omega$ which has most direct physical interpretation, whereas $\omega(p,p')$ encodes also commutation rules of all translates of $T,T'$. Algebraically $\omega$ is much more convenient to work with, essentially because the scalar part map $R \to \ZZ_n$ is not a homomorphism of $R$-modules. Sesquilinearity of $\omega$ implies that $\omega(x^\lambda p, x^\lambda p') = \omega(p,p')$,\ which is the statement that commutation relations of Pauli operators are translationally invariant.

Stabilizer code is a collection of eigenequations for a state\footnote{Here we regard $\Psi$ as a vector in some Hilbert space on which Pauli operators act. This Hilbert space is not specified a priori. However, one can reinterpret the eigenequations as $\Psi^{\mathrm{pre}}(T)=1$, where $\Psi^{\mathrm{pre}}$ is a state on the algebra of local operators. The Hilbert space and $\Psi$ may be then constructed from $\Psi^{\mathrm{pre}}$ using the GNS construction.} $\Psi$ of the form
\begin{equation}
    T \Psi = \Psi,
    \label{eq:stabilizer}
\end{equation}
where $T$ are Pauli operators (with phase factors chosen so that $1$ is in the spectrum of $T$). If such equations are imposed for two operators $T,T'$, then existence of solutions requires that $p,p' \in P$ satisfy $\omega(p,p')_0 =0$. In a translationally invariant code, the same condition has to be satisfied for all translates of $T,T'$, i.e. $\omega(p,p')=0$. It follows that the images in $P$ of operators defining the code generate a submodule $L$ with $\omega|_L\equiv 0$. Such submodules of $(P,\omega)$ are called isotropic. The stabilizer code determines a unique state if $L$ is Lagrangian, i.e. it is isotropic and every $p \in P$ such that $\omega(p,p')=0$ for every $p' \in L$ is in $L$. Throughout the article, we refer to codes with this property as \textit{Lagrangian codes}.

In quantum computation, one wishes to use spaces of states satisfying \eqref{eq:stabilizer} to store and protect information. When error occurs, there are violations of eigenequations called syndromes. On the other hand, one may also think of solutions of \eqref{eq:stabilizer} as ground states of a certain Hamiltonian. Then syndromes are also regarded as energetic excitations. Excited states are described by
\begin{equation} \label{excitation}
    T \Psi = \e^{\frac{2 \pi \i}{n} \varphi(p)} \Psi,
\end{equation}
where $p \in L$ corresponds to $T$ and $\varphi$ is a $\ZZ_n$-linear functional. The excitation is local (supported in a finite region) if $\varphi(x^{\lambda} p)$ vanishes for all but finitely many $\lambda \in \Lambda$.

The discussion above establishes a correspondence between a translationally invariant Pauli stabilizer code and an isotropic submodule of $(P,\omega)$. This correspondence will allow us to tap into the power of homological algebra. For the rest of the section, we develop the right hand side of the correspondence with additional generality.

\begin{defn}
For a $\ZZ_n$-module $M$, let $M^{\#} = \Hom_{\ZZ_n}(M, \ZZ_n)$. If $M$ is an $R$-module, then $M^{\#}$ is made an $R$-module as follows: $r \varphi(m) = \varphi(rm)$ for $r \in R$, $\varphi \in M^{\#}$ and $m \in M$. Moreover, we can define
\begin{equation}
 M_\Lambda^{\#} = \{ \varphi \in M^{\#} \, | \, \forall m \in M \ \  \varphi(x^\lambda m) \neq 0 \text{ for finitely many } \lambda \in \Lambda  \}.
\end{equation}
\end{defn}

\begin{defn}
Let $M$ be an $R$-module. We define $\overline M$ to be the $R$-module which coincides with $M$ as an abelian group, but with antipode $R$-action. In other words, if $m \in M$, we~denote the corresponding element of $\overline M$ by $\overline m$ and put
    $x^\lambda \overline m = \overline{x^{- \lambda} m}$.
Furthermore, we let $M^* = \Hom_R (\overline M ,R)$. $M^*$ is identified with the module of $\ZZ_n$-linear maps $f : M \to R$ such that $f (rm) = \overline r f(m)$ for $r \in R$ and $m \in M$.
\end{defn}

The following Lemma provides a useful description of $M^*$.

\begin{lem} \label{lem:dual_description}
Let $M$ be an $R$-module. The map taking $\varphi \in M^*$ to its scalar part $\varphi_0 \in \overline M^{\#}_{\Lambda}$ (i.e. $\varphi_0(m) = \varphi(m)_0$ for $m \in M$) is an $R$-module isomorphism $M^* \cong \overline M^{\#}_{\Lambda}$ with inverse given by the formula
\begin{equation}
    \varphi(m) = \sum_{\lambda \in \Lambda} \varphi_0(x^\lambda m) x^\lambda.
\end{equation}
\end{lem}

\begin{defn}
We denote the total ring of fractions of $R$ by $K$.
\end{defn}

Please see Appendix \ref{sec:Gorenstein} for some definitions referred to below.

\begin{lem} \label{lem:RGor}
$R$ is a Gorenstein ring of dimension $\mathrm{rk}(\Lambda)$, the free rank of $\Lambda$. Its total ring of fractions $K$~is a QF ring.
\end{lem}
\begin{proof}
$\ZZ_n$ is a QF ring by Baer's test. Thus $(-)^{\#}$ is an exact functor and $R^{\#}$ is an injective $R$-module, as $\Hom_R(-,R^{\#}) = (-)^{\#}$. Now suppose that $\mathrm{rk}(\Lambda)=0$. Then $R$ is finite, so $\dim(R)=0$. We have a~bilinear form
\begin{equation}
    R \times R \ni (r,r') \mapsto (rr')_0 \in \ZZ_n
\end{equation}
which yields an isomorphism $R \cong R^{\#}$. Hence $R$ is a QF ring. 

Next, let $\Lambda$ be arbitrary. We can split $\Lambda = \Lambda_1 \oplus \Lambda_2$, where $\Lambda_1$ is finite and $\Lambda_2$ free. We have $R = \ZZ_n[\Lambda_1][\Lambda_2]$, which is a Laurent polynomial ring in $D = \mathrm{rk}(\Lambda)$ variables over the QF ring $\ZZ_n[\Lambda_1]$. By Lemmas \ref{lem:Gor_loc_dim}, \ref{lem:Gor_pol}, $R$ is a~Gorenstein ring. Standard dimension theory shows that $\dim(R)=D$.

Invoking Lemma \ref{lem:Gor_loc_dim}, $K$ is also Goreinstein. It remains to show that $\dim(K)=0$. As $\ZZ_n[\Lambda_1]$ is Artinian, it is the product $\prod_{i=1}^s A_i$ of some Artin local rings $A_i$. Thus $R = \prod_{i=1}^s A_i [\Lambda_2]$. An element of $R$ is a zero-divisor if and only if its component in some $A_i[\Lambda_2]$ is a zero divisor, so $K = \prod_{i=1}^s K_i$, where $K_i$ is the total ring of fractions of $A_i[\Lambda_2]$. We will show that each $K_i$ is Artinian. 

Let $\mathfrak m_i$ be the maximal ideal of $A_i$. Then $\mathfrak m_i$ is nilpotent and every element of $A_i \setminus \mathfrak m_i$ is a unit. Clearly $\mathfrak m_i[\Lambda_2]$ is a prime ideal in $A_i[\Lambda_2]$. We claim that it is the unique minimal prime. Indeed, if $\mathfrak q \subset A_i[\Lambda_2]$ is a prime ideal, then $\mathfrak q \cap A_i$ is prime in $A_i$, hence equal to $\mathfrak m_i$. Thus $\mathfrak m_i[\Lambda] \subset \mathfrak q$ and the claim is established. Next, McCoy theorem \cite{McCoy} and nilpotence of $\mathfrak m_i$ imply that $\mathfrak m_i[\Lambda_2]$ is the set of zero divisors of $A_i[\Lambda_2]$, so every non-minimal prime ideal of $A_i[\Lambda_2]$ is killed in $K_i$. 
\end{proof}

Recall that an element of $R$ is said to be regular if it is not a zero divisor. The torsion submodule of an $R$-module M is the set of all elements of $M$ annihilated by a regular element of $R$. Equivalently, it is the kernel of the natural map $M \to M \otimes_R K$. If $M$ coincides with its torsion submodule, it is called a torsion module. If the torsion submodule of $M$ is $0$, then $M$ is said to be torsion-free. Quotient of any module by its torsion submodule is torsion-free.

\begin{lem} \label{lem:torsion_criteria}
Let $M$ be an $R$-module.
\begin{enumerate}
    \item $M$ is torsion if and only if $M^* = 0$.
    \item If $M$ is finitely generated, then $M$ is torsion-free if and only if it can be embedded in some free module
    $R^t$.
\end{enumerate}
\end{lem}
\begin{proof}
1. $\impliedby$: Let $M^* =0$. Then $ \Hom_K(\overline M \otimes_R K, K) = M^* \otimes_R K = 0$ (since $\Hom$ commutes with localization), so $\overline M \otimes_R K=0$ by Lemmas \ref{lem:RGor}, \ref{lem:QF_mods}. Thus $\overline M$, and hence $M$, is a torsion module.

2. $\implies$: As $M$ is torsion-free, it embeds in $M \otimes_R K$, which in turn embeds in $K^t$ by Lemmas \ref{lem:RGor} and \ref{lem:QF_mods}. Let $e_1, \dots, e_t$ be a basis of $K^t$. Since $M$ is finitely generated, there exists a regular element $d \in R$ such that the image of $M$ in $K^t$ is contained in the $R$-linear span of $d^{-1} e_1, \dots, d^{-1} e_t$, which is $R$-free.
\end{proof}

\begin{defn} \label{def:presymp}
Quasi-symplectic module is a finitely generated $R$-module $M$ equipped with a $\ZZ_n$-bilinear pairing $\omega : M \times M \to R$ satisfying
\begin{enumerate}
    \item $\omega(m',rm) = r  \omega(m',m)=\omega(\overline{r}m',m)$ for $r \in R$ and $m,m' \in M$,
    \item $\omega(m,m)_0 = 0$ for every $m \in M$,
    \item the map $\flat :  M \ni m \mapsto \omega( \cdot , m ) \in M^*$ is injective.
\end{enumerate}
We write $M^* / M $ for the quotient of $M^*$ by the image of $\flat$. If $M^* / M =0$, i.e.\ $\flat$ is an isomorphism, $(M , \omega)$ is called a symplectic module. If $N$ is another quasi-symplectic module, an isomorphism $f : M \to N$ is said to be symplectic if $\omega(f(m),f(m')) = \omega(m,m')$ for every $m,m' \in M$.
\end{defn}

\begin{prop} \label{prop:quasi_symp}
Let $M$ be a quasi-symplectic module. Then
\begin{enumerate}
    \item $M$ is torsion-free.
    \item For every $m,m' \in M$ we have $\omega(m,m') = - \overline{\omega(m',m)}$.
    \item $M^* / M$ is a torsion module. More generally, if~$N \subset M$ is a submodule, the cokernel of $M \ni m \mapsto \left. \omega( \cdot , m ) \right|_N \in N^*$ is a torsion module. 
\end{enumerate}
\end{prop}
\begin{proof}
1. If $m \in M$ is a torsion element, then $m \in \ker(\flat) = 0$.

2. For $r \in R$, let $r_\lambda$ be the coefficient of $x^\lambda \in R$. One has $r_\lambda = (r x^{- \lambda})_0$.
Plugging into $\omega(m,m)_0=0$ an element $m=m'+m''$ gives
\begin{equation}
\omega(m',m'')_0 =- \omega(m'',m')_0.
\end{equation}
Taking $m'' = x^\lambda m$ yields
    $\omega(m',m)_{-\lambda} = - \omega(m,m')_\lambda$,
establishing the claim.

3. For this part, we denote the functor $\Hom_K(-,K)$ by $(-)^\vee$. We have a~short exact sequence
\begin{equation}
    0 \to M \otimes_R K \xrightarrow{\flat'} M^* \otimes_R K \to (M^*/M) \otimes_R K \to 0,
    \label{eq:MstarM_quotient_frac}
\end{equation}
where $\flat' = \flat \otimes_R \id_K$. We may identify $M^* \otimes_R K$ with $(\overline M \otimes_R K)^\vee$, since Hom commutes with localization. As $\flat'$ is injective, the homomorphism
\begin{equation}
\flat'' = \Hom_K (\flat',K) : (\overline M \otimes_R K)^{\vee \vee}   \to (M \otimes_R K)^\vee
\end{equation}
is surjective by Lemma \ref{lem:RGor}. We identify $(\overline M \otimes_R K)^{\vee \vee} = \overline M \otimes_R K$, by Lemma~\ref{lem:QF_mods}. Using 2. we find that for any $m,m' \in M$ and $k,k' \in K$:
\begin{equation}
    \flat''(\overline m \otimes k ) (m' \otimes k') = - \overline{\flat'(m \otimes k) (\overline{m'} \otimes k') }.
\end{equation}
It follows at once that also $\flat'$ is surjective. Thus the short exact sequence \eqref{eq:MstarM_quotient_frac} yields $(M^* / M) \otimes_R K =0 $, i.e. $M^*/M$ is a torsion module.

Now let $N \subset M$ be a submodule. We have a short exact sequence
\begin{equation}
    0 \to N \to M \to M/N \to 0.
\end{equation}
Applying $*$ gives
\begin{equation}
    0 \to (M/N)^* \to M^* \to N^* \to \Ext^1_R(\overline{M/N},R).
\end{equation}
We have $\Ext^1_R(\overline{M/N},R) \otimes_R K = \Ext^1_K(\overline{M/N} \otimes_R K,K) =0$, since $\Ext$ commutes with localization. Hence $\Ext^1_R(\overline{M/N},R)$ is a torsion module. As both homomorphisms $M \to M^*$ and $M^* \to N^*$ have torsion cokernel, so does their composition. 
\end{proof}

\begin{cor} \label{cor:dim0_qsymp}
Suppose that $\Lambda$ is finite and let $M$ be a quasi-symplectic $R$-module. Then $M$ is symplectic. More generally, if $N \subset M$ is a submodule, then the map $M \ni m \mapsto \left. \omega(\cdot , m ) \right|_N \in N^*$ is surjective.
\end{cor}
\begin{proof}
The assumption guarantees that $R$ is a finite ring, so every element is either a zero-divisor or invertible. Hence torsion modules vanish.
\end{proof}

\begin{defn}
Let $M$ be a quasi-symplectic module and $N \subset M$ a submodule. Set $N^\omega = \{ m \in M \, | \, \left. \omega(\cdot , m) \right|_N =0 \}$. $N$ is called isotropic (resp. Lagrangian) if $N \subset N^\omega$ (resp. $N = N^\omega$).
\end{defn}

Recall that the saturation $\mathrm{sat}_M(N)$ of a submodule $N \subset M$ is defined to be the module of all $m \in M$ such that $rm \in N$ for some regular element $r \in R$. If $N = \mathrm{sat}_M(N)$, then $N$ is said to be saturated (in $M$). This is equivalent to $M/N$ being torsion-free. 

\begin{prop} \label{prop:perp_properties}
Let $N$ be a submodule of a quasi-symplectic module $M$.
\begin{enumerate}
    \item If $L \subset N$, then $N^\omega \subset L^\omega$.
    \item $N^\omega = N^{\omega \omega \omega}$.
    \item If $N$ is isotropic, $N \subset N^{\omega \omega} \subset N^\omega$, with equalities if $N$ is Lagrangian.
    \item $N^{\omega \omega} = \mathrm{sat}_M(N)$.
    \item $N^{\omega \omega} / N $ is a torsion module.
\end{enumerate}
\end{prop}
\begin{proof}
Points 1.-3. are established with simple manipulations. 

4. Clearly $\mathrm{sat}_M(N) \subset N^{\omega \omega}$. For the reverse inclusion, it is sufficient to check that if $N$ is saturated then $N^{\omega \omega} \subset N$. Let $m \in M \setminus N$. We will construct $z \in N^\omega$ such that $\omega(m,z) \neq 0$, showing that $m \not \in N^{\omega \omega}$. 

Put $L=N +Rm$. As $N$ is saturated, $L/N$ is torsion-free. Hence by Lemma \ref{lem:torsion_criteria} we have $(L/N)^* \neq 0$. Choose a nonzero element $\varphi \in (L/N)^*$. Composing with the quotient map $L \to L/N$ we obtain $\varphi' \in L^*$ which annihilates $N$ and $\varphi'(m) \neq 0$. By Proposition \ref{prop:quasi_symp} there exists a regular element $r \in R$ and $z \in M$ such that $r \varphi' = \left. \omega( \cdot , z ) \right|_{L} $. The element $z$ is as desired.

5. follows immediately from 4. and the definition of $\mathrm{sat}_M(N)$.
\end{proof}

\begin{cor} \label{cor:zero_dim_omom}
Suppose that $\Lambda$ is finite and let $M$ be a quasi-symplectic $R$-module. Then for every submodule $N \subset M$ we have $N^{\omega \omega} = N$.
\end{cor}
\begin{proof}
As in Corollary \ref{cor:dim0_qsymp}.
\end{proof}

\begin{prop}
Let $M$ be a quasi-symplectic module and $N \subset M$ an isotropic submodule.
\begin{enumerate}
    \item $N^{\omega \omega}/N$ is the torsion module of $N^\omega/N$.
    \item There exists an induced quasi-symplectic module structure on $N^\omega/N^{\omega \omega}$.
    \item There exists a canonical embedding $M/N^\omega \to N^*$ with torsion cokernel.
    \item There exists a canonical embedding $N^\omega \to (M/N)^*$ with torsion cokernel. If $M$ is symplectic, this embedding is an isomorphism.
\end{enumerate}
\end{prop}
\begin{proof}
1. follows from Proposition \ref{prop:perp_properties}. The bilinear form $\omega$ on $M$ restricted to $N^\omega$ has kernel $N^{\omega \omega}$, which establishes 2. By Proposition \ref{prop:quasi_symp}, we have a map $M \to N^*$ with torsion cokernel. Its kernel is clearly $N^\omega$, proving 3. 

4. Dualizing the short exact sequence $0 \to N \to M \to M/N \to 0$ gives
\begin{equation}
    0 \to (M/N)^* \to M^* \to N^* ,
\end{equation}
so $(M/N)^*$ may be identified with the set of $\varphi \in M^*$ with trivial restriction to $N$. Next we note that $\flat(N^\omega) = \flat(M) \cap (M/N)^*$, so
\begin{equation}
    (M/N)^* / \flat(N^\omega) = (M/N)^* / \left( \flat(M) \cap (M/N)^* \right) \subset M^* / M.
\end{equation}
\end{proof}

\begin{defn} \label{def:qf}
Let $M$ be an $R$-module. We say that $M$ is quasi-free if there exists a $\ZZ_n$-module $M_0$ such that $M \cong M_0 \otimes_{\ZZ_n} R$. We will also interpret elements of $M_0 \otimes_{\ZZ_n} R$ as polynomials in $x^\lambda$ with coefficients in $M_0$, thus writing $M_0 \otimes_{\ZZ_n} R = M_0 [\Lambda]$. 
\end{defn}

\begin{rmk}
Let $M, M_0$ be as in Definition \ref{def:qf}. Then $M_0$ is determined by $M$ up to isomorphism. $M$ is finitely generated over $R$ if and only if $M_0$ is finitely generated over $\ZZ_n$. Moreover, an $R$-module $M$ is free if and only if it is quasi-free and free as a $\ZZ_n$-module.
\end{rmk}

\begin{prop} \label{prop:standard_symp}
Let $P_0$ be a finitely generated $\ZZ_n$-module equipped with a~bilinear form $\omega_0 : P_0 \times P_0 \to \ZZ_n$ which~is
\begin{itemize}
    \item alternating: $\omega_0(p_0,p_0)=0$ for every $p_0 \in P_0$,
    \item nondegenerate: $\omega_0(\cdot, p_0) =0$ implies $p_0 = 0$.
\end{itemize}
Let $P = P_0[\Lambda]$ and define a $\ZZ_n$-bilinear form $\omega : P \times P \to R$ by
\begin{equation}
    \omega(p_0 x^\lambda , p_0' x^{\mu}) = \omega_0(p_0,p_0') x^{\mu - \lambda}, \qquad \text{for }  p_0, p_0' \in P_0, \, \lambda, \mu \in \Lambda.
\end{equation}
Then $(P, \omega)$ is a symplectic module.
\end{prop}
\begin{proof}
First note that $P_0$ and $P_0^{\#}$ have the same number of elements. Thus the map $P_0 \ni p_0 \mapsto \omega_0(\cdot, p_0) \in P_0^{\#}$, being injective by definition, is bijective. 

A short calculation shows that conditions 1. and 2. in the Definition \ref{def:presymp} are satisfied. Using the description of $P^*$ in Lemma \ref{lem:dual_description}, it is easy to see that $\flat$ is an isomorphism.
\end{proof}

Physically, $P_0$ is the group generated by clock and shift matrices acting on qubits on a single lattice site, considered modulo phases.

\begin{defn} \label{def:codes}
A stabilizer code is a tuple $\mathfrak C = (n, \Lambda, L,P)$, where $P$ is symplectic module over $R = \ZZ_n[\Lambda]$ as constructed in Proposition \ref{prop:standard_symp} and $L \subset P$ is an isotropic submodule. We will also abbreviate $\mathfrak C = ( \Lambda,L,P)$ or $(L,P)$ when there is no danger of confusion. To $\mathfrak C$ we associate
\begin{itemize}
    \item integer dimension $D = \mathrm{rk}(\Lambda)$, the free rank of $\Lambda$,
    \item quasi-symplectic module $S(\mathfrak C) = L^\omega/ L^{\omega \omega}$,
    \item torsion module $Z(\mathfrak C) = L^{\omega \omega} /L$,
    \item torsion module $Q(\mathfrak C) = L^* / (P/L^\omega)$.
\end{itemize}
We say that $\mathfrak C$ is saturated if $L \subset P$ is saturated ($Z(\mathfrak C)=0$) and Lagrangian if $L \subset P$ is Lagrangian ($Z(\mathfrak C) = S(\mathfrak C) = 0$). An isomorphism of stabilizer codes $(L,P) \to (L',P')$ is a symplectic isomorphism $P \to P'$ taking $L$ to $L'$.
\end{defn}

Let us interpret physically objects defined above. Let $\mathcal H$ be a Hilbert space on which local Pauli operators act irreducibly and let $\mathcal H_0 \subset \mathcal H$ be the space of solutions of \eqref{eq:stabilizer} in $\mathcal H$. We assume that $\mathcal H_0 \neq 0$. One can show that operators in $L^{\omega}$ act irreducibly in $\mathcal H_0$. Since they commute with operators in $L^{\omega \omega}$, the latter act in $\mathcal H_0$ as scalars. This is a~trivial statement for operators in $L$, but for operators in $L^{\omega \omega} \setminus L$ the conclusion relies on the irreducibility of $\mathcal H$, through Schur's lemma. Values of the latter operators may be changed by acting on a state with a suitable automorphism of the local operator algebra (more precisely, a non-local Pauli operator) which preserves all operators in $L$. This gives a state which is not representable by an element of $\mathcal H$ (belongs to a~different superselection sector). Hence we have the following interpretations.
\begin{itemize}
    \item $Z(\mathfrak C)$ labels order parameters for spontaneously broken symmetries. 
    If $L^{\omega \omega}$ is Lagrangian, isomorphism classes of representations $\mathcal H$ with $\mathcal H_0 \neq 0$ are in bijection with $Z(\mathfrak C)^{\#}$ (and hence also with $Z(\mathfrak C)$ if $Z(\mathfrak C)$ is finite).
    \item Elements of $S(\mathfrak C)$ are Pauli operators acting in $\mathcal H_0$ (sometimes called logical operators) modulo operators which act in $\mathcal H_0$ as scalars. Hence $\dim(\mathcal H_0)$ is the square root of the number of elements\footnote{$S(\mathfrak C)$ is at most countably infinite. It is is not finite, $\dim(\mathcal H_0)$ in this statement has to be interpreted as the Hilbert dimension, not the algebraic dimension.} of $S(\mathfrak C)$. 
\end{itemize}
By the discussion around \eqref{excitation} and Lemma \ref{lem:dual_description}, module $L^*$ parametrizes local excitations. Therefore $Q(\mathfrak C)= L^* / (P/L^\omega)$ is the module of local excitations modulo excitations which can be created by acting with local operators. 

\section{Topological charges}

In this section we define a series of homological invariants $Q^i$, with $Q^0$ isomorphic to $Q$ in Definition \ref{def:codes}. Moreover, we show that $Q^0$ is isomorphic to the module of topological point excitations defined in \cite{haah2013commuting} and derive some general properties of $Q^i$. Firstly, we show that $Q^i(\mathfrak C)=0 $ for $i > D-1$ (and also for $i = D-1$ for saturated codes). Secondly, we obtain bounds on Krull dimensions of $Q^i(\mathfrak C)$. We expect $Q^i$ to describe $i$-dimensional excitations (or defects). This is shown in Section 6 for Lagrangian codes such that all $Q^i$ have Krull dimension zero. Computations of $Q^i$ for certain specific codes are presented in Section 7.

We remark that it follows immediately from our results that for saturated codes $\mathfrak C$ with $D=2$, the module $Q(\mathfrak C)$ either vanishes or has Krull dimension zero. Together with the discussion in Section 5 it implies that all point excitations are mobile, i.e.\ they can be transported around by suitable string operators. This result has previously been shown only for codes with qudits of prime dimension \cite{haah2013commuting}. Method adapted therein does not generalize to the case of composite qudit dimension due to the failure of Hilbert's syzygy theorem, a~crucial ingredient of the proof. 

\begin{lem} \label{lem:vanishing_Ext}
If $M$ is a quasi-free module and $N$ is free over $\ZZ_n$, then for $i> 0$ 
\begin{equation}
    \Ext^i_R(M,N) = 0 , \qquad \Tor_i^R(M, N) = 0.
\end{equation}
\end{lem}
\begin{proof}
Every $\ZZ_n$-module is a direct sum of cyclic modules, so without loss of generality $M = \ZZ_k[\Lambda] $ with $k|n$. Let $l = \frac{n}{k}$. We have a~free resolution
\begin{equation}
 \dots \to R \xrightarrow{k} R \xrightarrow{l} R \xrightarrow{k} R \xrightarrow{\text{mod } k} M \to 0.
\end{equation}
Erasing $M$ and applying $\mathrm{Hom}_R(-, N)$ we obtain the sequence
\begin{equation}
0 \to N \xrightarrow{k} N \xrightarrow{l} N \to \dots,
\end{equation}
which is exact in every degree $i>0$. This establishes the claim for $\Ext$. The argument for $\Tor$ is analogous.
\end{proof}

\begin{prop} \label{prop:QExt}
Let $\mathfrak C=(L,P)$ be a stabilizer code. We have
\begin{equation}
    Q(\mathfrak C) \cong \Ext^1_R (\overline{P/L},R).
\end{equation}
\end{prop}
\begin{proof}
Consider the short exact sequence
\begin{equation}
    0 \to L \to P \to P/L \to 0.
    \label{eq:LPPL_seq}
\end{equation}
We apply $*$, use Lemma \ref{lem:vanishing_Ext} and identify $(P/L)^* = L^\omega$, $P^* = P$ to get
\begin{equation}
    0 \to L^\omega \to P \to L^* \to \Ext^1_R(\overline{P/L},R) \to 0,
\end{equation}
so $\Ext^1_R(\overline{P/L},R) \cong L^* /(P/L^\omega) = Q(\mathfrak C)$.
\end{proof}

Proposition \ref{prop:QExt} motivates the definition of generalized charge modules.

\begin{defn}
Generalized charge modules of a stabilizer code $\mathfrak C=(L,P)$ are defined as 
\begin{equation}
    Q^i(\mathfrak C) = \Ext^{i+1}_R(\overline{P/L},R), \qquad i \geq 0.
\end{equation}
\end{defn}

\begin{prop}
For $i>0$ we have a canonical isomorphism
\begin{equation}
    Q^i(\mathfrak C) \cong \Ext^{i}_R(\overline L,R).
\end{equation}
\end{prop}
\begin{proof}
Inspect the long exact sequence obtained by applying $(-)^*$ to \eqref{eq:LPPL_seq}.
\end{proof}

The next proposition shows that our definition of $Q(\mathfrak C)$ agrees with topological point excitations in \cite{haah2013commuting}. 

\begin{prop} \label{prop:Q_Haah}
Let $(L,P)$ be a stabilizer code and let $\sigma : F \to P$ be a homomorphism with $F$ quasi-free and $\im(\sigma) = L$. Let $T$ be the torsion submodule of the cokernel of $\sigma^* : P^* \to F^*$. Then $T \cong Q^0(\mathfrak C)$.
\end{prop}
\begin{proof}
Choose a quasi-free module $F'$ and a homomorphism $\iota : F' \to F$ with image $\ker(\sigma)$. One may extend it to a quasi-free resolution of $P/L$:
\begin{equation}
    \dots \to F' \xrightarrow{\iota} F \xrightarrow{\sigma} P \to P/L \to 0.
\end{equation}
By Lemma \ref{lem:vanishing_Ext} this resolution may be used to compute $\Ext^\bullet(\overline{P/L},R)$. Thus we erase $P/L$ and apply $(-)^*$, yielding the complex
\begin{equation}
    0 \to P^* \xrightarrow{\sigma^*} F^* \xrightarrow{\iota^*} F'^* \to \dots
\end{equation}
whose homology $\ker(\iota^*) / \im(\sigma^*)$ in degree $1$ is $Q^0(\mathfrak C)$. This exhibits $Q^0(\mathfrak C)$ as a~submodule of $\coker(\sigma^*)$. It is contained in $T$ because $Q^0(\mathfrak C)$ is torsion. It only remains to show that every $\varphi \in F^*$ representing an element of $T$ is in $\ker(\iota^*)$. Indeed, let $r \varphi = \sigma^* (\psi)$ for some $r \in R$ not a zero-divisor and $\psi \in P^*$. Then $r \, \iota^*(\varphi) = 0$, so $\iota^*(\varphi )=0$ since $F'^*$ is torsion-free.
\end{proof}

Recall that the dimension $\dim(M)$ of an $R$-module $M$ is defined as the Krull dimension of the quotient ring $R / \ann(M)$, where $\ann(M)$ is the annihilator of $M$. A nonzero module has a nonnegative Krull dimensions. By convention, the zero module has Krull dimension $-\infty$.

\begin{prop} \label{prop:Q_dim}
Let $\mathfrak C$ be a stabilizer code.
\begin{enumerate}
    \item $Q^i(\mathfrak C) =0 $ for $i \geq D$.
    \item $Q^{D-1}(\mathfrak C) \cong \Ext^D_R(\overline{Z(\mathfrak C)},R)$. In particular $Q^{D-1}(\mathfrak C) =0$ if $\mathfrak C$ is saturated.
    \item $ \dim(Q^i(\mathfrak C)) \leq D-1-i$. In particular $Q^i(\mathfrak C)$ is a torsion module.
    \item If $\mathfrak C$ is saturated, then $\dim(Q^i(\mathfrak C)) \leq D-2-i$.
\end{enumerate}
In particular, saturated 1D codes have no topological charges.
\end{prop}
\begin{proof}
1. follows from the definition of a Gorenstein ring. 3. follows from Lemma \ref{lem:dim_Ext_bound}. Now suppose that $\mathfrak C$ is saturated. Then $P/L$ is torsion-free, so by Lemma \ref{lem:torsion_criteria} there exists a short exact sequence
\begin{equation}
    0 \to P/L \to F \to M \to 0
\end{equation}
with $F$ finite free. Applying $(-)^*$ gives a long exact sequence from which 
\begin{equation}
    \Ext^{i+1}_R(\overline {P/L},R) \cong \Ext^{i+2}_R(\overline M,R), \qquad i \geq 0.
\end{equation}
In particular $\Ext^D_R(\overline{P/L},R)=0$. Invoking Lemma \ref{lem:dim_Ext_bound} establishes 4.

2. We have a short exact sequence
\begin{equation}
    0 \to L^{\omega \omega} / L \to P/L \to P / L^{\omega \omega} \to 0.
\end{equation}
Apply $*$ and use $\Ext^D_R(\overline{P/L^{\omega \omega}},R) =0$, established in the proof of 4.
\end{proof}

\section{Operations on Pauli stabilizer codes}

One Pauli stabilizer code may give rise to various other codes. For example, one may ``compatify'' some (even all) spatial directions, i.e.\ replace $\Lambda$ by a~quotient group. Another possibility is stacking of infinitely many copies of a certain code to create a code with higher dimension. Finally, one has coarse-graining, which does not change the code, but forgets about some of its translation symmetry. In this section we discuss stacking and coarse-graining (in particular how they affect invariants of a code), but compactifications are postponed to future work. Moreover, we explain that the choice of $n$ (which has to be a common multiple of qubit dimensions) does not matter and that the whole theory reduces to the case when $n$ is a prime power.

\begin{defn}
Let $\mathfrak C = (n, \Lambda,L,P)$ be a stabilizer code and let $k$ be a~positive integer divisible by $n$. Then we may regard $L$ and $P$ as $\ZZ_k[\Lambda]$-modules, yielding a stabilizer code $\mathfrak C'=(k,\Lambda, L,P)$. We will not distinguish between $\mathfrak C$ and $\mathfrak C'$. Proposition below shows that this does not affect charge codes. Given data $(\Lambda, L,P)$ we choose $n$ (needed to define the ring $R$) as the smallest positive integer annihilating the abelian group $P$.
\end{defn}

\begin{prop}
Let $\mathfrak C$, $\mathfrak C'$ be as above. Then $S(\mathfrak C')$ coincides with $S(\mathfrak C)$ regarded as a~$\ZZ_k[\Lambda]$-module. Similarly, $Z(\mathfrak C) = Z(\mathfrak C')$ and $Q^i(\mathfrak C') = Q^i(\mathfrak C)$.
\end{prop}
\begin{proof}
If $M \subset P$ is a submodule, $M^{\omega}$ is the same over $\ZZ_n[\Lambda]$ and $\ZZ_k[\Lambda]$. This establishes the first two equalities. For the last one, note that $\Ext^\bullet_R( - , R)$ may be computed using quasi-free resolutions by Lemma \ref{lem:vanishing_Ext}, a~quasi-free $\ZZ_n[\Lambda]$-module is also quasi-free over $\ZZ_k[\Lambda]$, and for any $\ZZ_n[\Lambda]$-module $M$ we have
\begin{equation}
 \Hom_{\ZZ_k[\Lambda]}(M,\ZZ_k[\Lambda]) \cong \Hom_{\ZZ_n[\Lambda]}(M,\ZZ_n[\Lambda]).   
\end{equation}
\end{proof}

\begin{defn}
Direct sum of stabilizer codes is defined by
\begin{equation}
    (n,\Lambda, L, P) \oplus (m,\Lambda, L',P') = (\gcd(n,m),\Lambda, L,P). 
\end{equation}
\end{defn}

Clearly $S(\mathfrak C), Z(\mathfrak C)$ and $Q^i(\mathfrak C)$ are additive.

\begin{prop} \label{prop:CRT}
Let $\mathfrak C= (n,\Lambda, L ,P)$ be a stabilizer code and let $n=\prod \limits_{i=1}^r p_i^{n_i}$ be the prime decomposition of $n$. Then
\begin{equation}
    \mathfrak C = \bigoplus_{i=1}^r (p_i^{n_i},\Lambda, L_i , P_i),
\end{equation}
where $P_i = \{ m \in P \, | \, p_i^{n_i} m =0 \}$, $L_i = L \cap P_i$.
\end{prop}
\begin{proof}
Chinese remainder theorem.
\end{proof}

Note that Proposition \ref{prop:CRT} implies that the study of stabilizer code with general $n$ reduces to the case when $n$ is a prime power.

\begin{defn} \label{def:stacking}
Let $\Gamma$ be a finitely generated abelian group and let $\iota : \Lambda \to \Gamma$ be a homomorphic embedding. For any $\ZZ_n[\Lambda]$-module $M$ let $\iota_* M = M \otimes_{\ZZ_n[\Lambda]} \ZZ_n[\Gamma]$. Then $\iota_*$ is an exact functor because $\ZZ_n[\Gamma]$ is free over $\ZZ_n[\Lambda]$. In particular for a~stabilizer code $(\Lambda, L, P)$ we have $\iota_* L \subset \iota_* P$, allowing us to define
\begin{equation}
\iota_* (\Lambda, L, P) = (\Gamma, \iota_* L, \iota_* P).
\end{equation}
\end{defn}

The operation introduced in Definition \ref{def:stacking} may be thought of as stacking of $\Gamma/\Lambda$ layers of the system described by $(\Lambda,L,P)$. Let us note that
\begin{equation}
    S(\iota_* \mathfrak C) = \iota_* S(\mathfrak C), \qquad Z(\iota_* \mathfrak C) = \iota_* Z(\mathfrak C), \qquad Q^i(\iota_* \mathfrak C) = \iota_* Q^i(\mathfrak C).
    \label{eq:charge_stacking}
\end{equation}
Due to these simple formulas, structure of charge modules may be used to show that a certain system cannot be obtained from a lower dimensional system by stacking. Here we note only a simple criterion based on whether charge modules vanish.

\begin{prop}
Suppose that $\mathfrak C$ is a stabilizer code with $Q^i(\mathfrak C) \neq 0$. Then $\mathfrak C$ is not isomorphic to any $\iota_* (\Lambda, L , P)$ with $\mathrm{rk}(\Lambda) < i+1$. If $\mathfrak C$ is saturated, $\mathrm{rk}(\Lambda) = i+1$ is also excluded.
\end{prop}
\begin{proof}
Formula \eqref{eq:charge_stacking} and Proposition \ref{prop:Q_dim}.
\end{proof}

\begin{prop}
Suppose that $\mathfrak C$ is a stabilizer code which is not saturated. Then $\mathfrak C$ is not isomorphic to any $\iota_* (\Lambda,L,P)$ with $\mathrm{rk}(\Lambda)=0$.
\end{prop}
\begin{proof}
Zero-dimensional systems have $Z(\mathfrak C)=0$ by Corollary \ref{cor:zero_dim_omom}. The claim follows from \eqref{eq:charge_stacking}.
\end{proof}

\begin{defn}
Let $\iota : \Gamma \to \Lambda$ be a finite index embedding. If $M$ is a $\ZZ_n[\Lambda]$ module, we let $\iota^* M$ be $M$ treated as $\ZZ_n[\Gamma]$-module. We define
\begin{equation}
    \iota^* (\Lambda, L, P) = (\Lambda, \iota^* L, \iota^* P).
\end{equation}
This operation is called coarse graining.
\end{defn}

\begin{prop}
Coarse graining satisfies
\begin{equation}
 S(\iota^* \mathfrak C) = \iota^* S(\mathfrak C), \qquad  Z(\iota^* \mathfrak C) = \iota^* Z(\mathfrak C), \qquad Q^i(\iota^* \mathfrak C) = \iota^* Q^i(\mathfrak C).
 \label{eq:cg_charges}
\end{equation}
\end{prop}
\begin{proof}
Let $M \subset P$ be a submodule and $p \in P$. Then $p \in M^{\omega}$ if and only if the scalar part of $\omega(m,p)$ vanishes. The scalar part is unchanged by coarse graining, so $(\iota^* M)^\omega = \iota^* (M^{\omega})$. This establishes first two equalities in \eqref{eq:cg_charges}. For the last one, $\iota^*$ is an exact functor which takes free modules to free modules and commutes with $(-)^*$, as one verifies using Lemma \ref{lem:dual_description}.
\end{proof}
 
\section{Mobility theorem}

A local excitation is said to be mobile if there exist local Pauli operators which `move' it in all non-compact directions of the lattice. By `move', we mean destroying the excitation and creating its displaced copy, without creating additional excitations.

Recall that $Q = L^*/(P/L^\omega)$ describes all local excitations modulo those creatable by local Pauli operators. According to the previous paragraph, an~excitation $e\in L^*$ can be displaced by an element $\gamma \in \Lambda$ if and only if $(x^\gamma - 1)e\in  P/L^\omega$. In conclusion, mobility of all local excitations is equivalent to the existence of a subgroup $\Gamma\subset \Lambda$ of finite index such that $x^\gamma - 1$ annihilates $Q$ for each $\gamma\in \Gamma$. We now show that this condition is also equivalent to the vanishing of the Krull dimension of $R$-module $Q$.

\begin{lem} \label{lem:divisibility}
If $n,r$ are positive integers, let $L_n(r)$ be the largest integer such that $(x-1)^{L_n(r)}$ divides $x^{n^r} -1$ in $\ZZ_n[x]$. For example, $L_p(r)= p^r$ for any prime number $p$. One has $\limsup \limits_{r \to \infty} L_n(r) = \infty$.  
\end{lem}
\begin{proof}
Factorization $x^n-1 = - (x-1)^2 \sum \limits_{j=0}^{n-1} (j+1) x^j$ implies that $L_n(1)\geq 2$. We will show that $L_n(2r) \geq L_n(r)^2$. Write $x^{n^r}-1 = (x-1)^{n^r} f(x)$. Then
\begin{align}
x^{n^{2r}}-1 &= (x^{n^r})^{n^r}-1 = (x^{n^r}-1)^{L_n(r)} f(x^{n^r}) \\ &=  (x-1)^{L_n(r)^2} f(x)^{L_n(r)} f(x^{n^r}). \nonumber
\end{align}
\end{proof}

\begin{prop} \label{prop:mob}
If $\mathfrak a\subset R$ is an ideal, then $\dim(R/\mathfrak a) =0$ if and only if there exists a subgroup $\Gamma \subset \Lambda$ of finite index such that $x^\gamma-1 \in \mathfrak a$ for every $\gamma \in \Gamma$.
\end{prop}
\begin{proof}
$\impliedby $ : $R/\mathfrak a$ is a finite ring, so $\dim(R/\mathfrak a)=0$.

$\implies$ : choose $\lambda_1, \dots, \lambda_D \in \Lambda$ which generate a subgroup of finite index and put $x_i = x^{\lambda_i}$. If $\mathfrak m \subset R$ is a maximal ideal, then $R / \mathfrak m$ is a finite field, so there exists a positive integer such that $x_i^r -1 \in \mathfrak m$. As $\dim(R/\mathfrak a) =0$, there exist finitely many maximal ideals $\mathfrak m \subset R$ containing $\mathfrak a$. Thus it is possible to choose $r$ such that
\begin{equation}
    (x_1^r - 1, \dots, x_D^r -1) \subset \bigcap_{\mathfrak{m}\supset \mathfrak a} \mathfrak{m}= \sqrt{\mathfrak a}.
\end{equation}
Since $R$ is Noetherian, $\sqrt{\mathfrak a}^N \subset \mathfrak a$ for large enough $N$. Lemma \ref{lem:divisibility} implies that there exist $N,L$ such that
\begin{equation}
 (x^L_1 - 1, \dots, x_D^L -1) \subset ((x_1^r - 1)^N, \dots, (x_D^r -1)^N)\subset \mathfrak a.  
\end{equation}
We may take $\Gamma$ to be the span of $L \lambda_1, \dots, L \lambda_D$.
\end{proof}

\begin{cor} \label{cor:2d_DIM}
All local excitations of a stabilizer code $\mathfrak C$ are mobile if and only if $\dim(Q(\mathfrak C))=0$. In particular this is true if $D = 1$ or $\mathfrak C$ is saturated and $D=2$.  
\end{cor}
\begin{proof}
The second part of the statement follows from dimension bounds in Proposition \ref{prop:Q_dim}.
\end{proof}

Though logically equivalent, the condition $\dim(R/\mathfrak a)=0$ avoids mentioning a~finite index subgroup of $\Lambda$. It is also the easier condition to establish in a~proof, due to the large number of results in dimension theory. An example is given by Corollary \ref{cor:2d_DIM} above.

A direct characterization of mobility for $i$-dimensional topological charges in $Q^i, \ i>0$ may be possible, given an interpretation of charges in terms of extended excitations. We leave this to future efforts. Instead we make the conjectural definition that mobility for $Q^i$ is still equivalent to $\dim (Q^i)=0$. We sometimes call a code $\mathfrak C$ mobile if $\dim (Q^i(\mathfrak C))=0$ for all $i$. In the next section we will see that under this assumption elements of $Q^i(\mathfrak C)$ may indeed be interpreted as excitations, which are mobile in a suitable sense.

\section{Codes with only mobile excitations}

This section is devoted to analysis of topological charges for mobile codes. Mobility allows to describe topological charges in terms of \v Cech cocycles. Cup product for \v Cech cohomology fits a physical process commonly known as braiding. It furnishes an algebraic description of exchange relations for mobile excitations. A direct physical interpretation of \v Cech cocycles is also given.

\subsection{Mathematical preliminaries}

If $A$ is a ring and $M$ an $A$-module, let $\E_A(M)$ be the injective envelope of $M$. We refer to Appendix \ref{sec:loc_coh} for other definitions and facts used below.

\begin{prop} \label{prop:tors_Rhash}
Let $\mathfrak a \subset R$ be an ideal such that $\dim(R/\mathfrak a) =0$. Then
\begin{equation}
    \Gamma_{\mathfrak a}(R^{\#}) \cong \bigoplus_{\mathfrak m} \E_R(R/\mathfrak m),
\end{equation}
the sum being taken over maximal ideals of $R$ containing $\mathfrak a$.
\end{prop}
\begin{proof}
Lemma \ref{lem:lc_codim_zero} allows to reduce to the case of $\mathfrak a$ being itself a maximal ideal $\mathfrak m$. We put $k = R / \mathfrak m$. $R$-module $R^{\#}$ represents the exact cofunctor $(-)^{\#}$ on the category of $R$-modules, so it is injective. By \cite[Proposition~3.88]{Lam}, $\Gamma_{\mathfrak m}(R^{\#})$ is also injective. It is easy to see that $k^{\#} \cong \{ \varphi \in R^{\#} \, | \, \mathfrak m \varphi =0 \}$ is an essential submodule of $\Gamma_{\mathfrak m}(R^{\#})$, so $\Gamma_{\mathfrak m}(R^{\#})=\E_R(k^{\#})$. The proof will be completed by showing that $k^{\#} \cong k$ as an $R$-module. As $k^{\#}$ is annihilated by $\mathfrak m$, it is a $k$-vector space. We have to argue that its dimension over $k$ is $1$. Let $p$ be the characteristic of $k$. Every element of $k^{\#}$ factors through $\ZZ_p$, so
\begin{equation}
    \dim_{\ZZ_p}(k^{\#}) = \dim_{\ZZ_p}(\Hom_{\ZZ_p}(k,\ZZ_p)) = \dim_{\ZZ_p}(k)
\end{equation}
and hence $\dim_k(k^{\#}) = \frac{\dim_{\ZZ_p}(k^{\#})}{\dim_{\ZZ_p}(k)}=1$.
\end{proof}

\begin{lem} \label{lem:height}
Every maximal ideal of $R$ has height $D$.
\end{lem}
\begin{proof}
$R$ is a product of rings $\ZZ_{p^t}[\Lambda]$ where $p$ is prime and $t \in \mathbb N$, so we may assume that $n=p^t$ with no loss of generality. Then $R$ is an extension of $S=\ZZ_{p}[\Lambda]$ by a nilpotent ideal, so its poset of prime ideals is isomorphic to that of $S$. The result for $S$ is standard, see e.g. \cite[Corollary 13.4]{Eisenbud}.
\end{proof}

\begin{prop} \label{prop:loc_ch_R}
Let $\mathfrak a \subset R$ be an ideal such that $\dim(R/\mathfrak a) =0$ and let $M$ be a quasi-free module. Then $\H^j_{\mathfrak a}(M)=0$ for $j \neq D$ and 
\begin{equation}
    \H^D_{\mathfrak a}(M) \cong  \left( \bigoplus_{\mathfrak m} \E_R(R/\mathfrak m) \right) \otimes M,
\end{equation}
the sum being taken over maximal ideals of $R$ containing $\mathfrak a$.
\end{prop}
\begin{proof}
Lemma \ref{lem:lc_codim_zero} allows to reduce to the case of $\mathfrak a$ being a maximal ideal $\mathfrak m$. First consider the case $M = R$. By maximality of $\mathfrak m$ and $\H^j_{\mathfrak m}(R)$ being $\mathfrak m$-torsion, every element of $R \setminus \mathfrak m$ acts as an invertible endomorphism of $\H^j_{\mathfrak m}(R)$. Thus we have $R$-module isomorphisms $\H^j_{\mathfrak m}(R) \cong \H^j_{\mathfrak m}(R)_{\mathfrak m} \cong \H^j_{ \mathfrak m' } (R_{\mathfrak m})$, where $\mathfrak m'$ is the extension of $\mathfrak m$ in $R_{\mathfrak m}$. The second isomorphism follows from Lemma \ref{lem:loc_ch_localization}. By Lemma \ref{lem:height}, $R_{\mathfrak m}$ is a Gorenstein ring of dimension $D$, so Lemma \ref{lem:Gor_loc_ch} gives $\H^j_{\mathfrak m}(R)=0$ for $j \neq D$ and $\H^D_{\mathfrak m}(R) \cong \E_{R_{\mathfrak m}}(R_{\mathfrak m}/\mathfrak m') \cong \E_R(R / \mathfrak m)$.

Local cohomology can be computed using the \v{C}ech complex, so the result for $M=R$ shows that $\Cech^\bullet(\mathbf t, R)$ is a flat resolution of $\E := \E_R(R / \mathfrak m)$, up to a~degree shift. Since $\Cech^\bullet(\mathbf t, M) \cong \Cech^\bullet(\mathbf t , R) \otimes_R M$, this implies that for any module $M$ we have $\CH^p(\mathbf t, M) \cong \Tor_{D - p}^R(E,M)$. Now specialize to the case of $M$ being quasi-free and invoke Lemma \ref{lem:vanishing_Ext}.
\end{proof}

Let $\Gamma \subset \Lambda$ be a subgroup such that $\Lambda / \Gamma$ is finite and let $\gamma_1 , \dots, \gamma_D$ be a~basis of $\Gamma$. We put 
\begin{equation}
 x_i = x^{\gamma_i}, \qquad t_i = 1-x^{\gamma_i}, \qquad \mathfrak a = (t_1, \dots ,t_D)
 \label{eq:Cech_setup}
\end{equation}
and consider the \v{C}ech complex $\Cech^\bullet(\mathbf t, R)$ (see Appendix \ref{sec:Cech}). Lemma \ref{lem:loc_ch_Cech} and Propositions \ref{prop:tors_Rhash}, \ref{prop:loc_ch_R} show that its only nonzero cohomology module $\CH^D(\mathbf t, R)$ is isomorphic to $\Gamma_{\mathfrak a}(R^{\#})$. Our next goal is to construct an explicit isomorphism.

\begin{defn} \label{def:formal_series}
Let $\ZZ_n[[\Lambda]]$ be the set of formal sums $\sum_{\lambda \in \Lambda} r_\lambda x^\lambda$. This is an abelian group, but in general not a ring: the product
\begin{equation}
    \left( \sum_{\lambda \in \Lambda} r_{\lambda} x^{\lambda}  \right) \left( \sum_{\mu \in \Lambda} r'_{\mu} x^{\mu} \right) = \sum_{\lambda \in \Lambda} \left( \sum_{\mu \in \Lambda} r_{\lambda - \mu} r'_{ \mu} \right) x^\lambda
\end{equation}
is well-defined only if for every $\lambda \in \Lambda$ there are only finitely many $\mu \in \Lambda$ such that both $r_{\lambda -\mu}$ and $r'_{ \mu}$ is nonzero. This condition is always satisfied if one of the two factors is in $R$, so $\ZZ_n[[\Lambda]]$ is an $R$-module. Using the pairing
\begin{equation}
    \ZZ_n[\Lambda] \times \ZZ_n[[\Lambda]] \ni (r,r') \mapsto (rr')_0 \in \ZZ_n,
\end{equation}
we identify $\ZZ_n[[\Lambda]]$ with $R^{\#}$.
\end{defn}

Recall that $\Cech^D(\mathbf t, R)=R_{t_1 \dots t_D}$ and that $\CH^D(\mathbf t, R)$ is the quotient of $R_{t_1 \dots t_D}$ by the sum of images of $R_{t_1 \dots t_{j-1} t_{j+1} \dots t_D}$ (module of coboundaries). 

\begin{defn} \label{def:Res}
We consider formal Laurent expansions of $\frac{1}{t_i}$ (regarded as elements of $R^{\#}$) into positive and negative powers of $x_i$:
\begin{equation}
    \left( \frac{1}{t_i} \right)_+ = \sum_{j=0}^\infty x_i^j, \qquad \left( \frac{1}{t_i} \right)_- = - \sum_{j=1}^\infty x_i^{-j},
\end{equation}
The residue homomorphism $\Res : R_{t_1 \dots t_D} \to R^{\#}$ is defined by
\begin{equation}
    \frac{r}{t_1^{k_1} \cdots t_D^{k_D}} \mapsto r \prod_{i=1}^D \left[ \left( \frac{1}{t_i} \right)^{k_i}_+ - \left( \frac{1}{t_i} \right)^{k_i}_- \right].
    \label{eq:Res_hom}
\end{equation}
This is well-defined because $t_i \left( \frac{1}{t_i} \right)_{\pm} = 1$. 
\end{defn}

\begin{prop}
$\ker(\Res)$ is the module of coboundaries and the image of $\Res$ is $\Gamma_{\mathfrak a}(R^{\#})$. Therefore $\Res$ induces an isomorphism $\CH^D(\mathbf t, R) \to \Gamma_{\mathfrak a}(R^{\#})$.
\end{prop}
\begin{proof}
A \v{C}ech coboundary is a sum of elements as on the left hand side of \eqref{eq:Res_hom} with at least one $k_i$ equal to zero, each of which is annihilated by $\Res$. Moreover, the right hand side of \eqref{eq:Res_hom} is annihilated by $t_i^{k_i}$, so it belongs to $\Gamma_{\mathfrak a}(R^{\#})$. We have obtained an induced homomorphism $ \CH^D(\mathbf t, R) \to \Gamma_{\mathfrak a}(R^{\#})$. From now on the symbol $\Res$ refers to this induced homomorphism. Let $z$ be the cohomology class of $\frac{1}{t_1 \dots t_D}$. Clearly $\mathfrak a \subset \ann(z)$. We evaluate
\begin{equation}
    \Res (z) = \sum_{\gamma \in \Lambda} x^\gamma.
\end{equation}
One checks that the annihilator of the right hand side is $\mathfrak a$, so $\ann(z) \subset \mathfrak a$. We deduce that the submodule $M$ of $\CH^D(\mathbf t, R)$ generated by $z$ intersects $\ker(\Res)$ trivially. Clearly $M$ is an essential submodule of $\CH^D(\mathbf t, R)$, so $\Res$ is injective. Propositions \ref{prop:tors_Rhash}, \ref{prop:loc_ch_R} imply that it is an isomorphism.
\end{proof}

\subsection{Physical interpretations of charges} \label{sec:operators}

For the rest of this section we assume that $\mathfrak C=(\Lambda, L,P)$ is a Lagrangian stabilizer code such that $\dim(Q^i(\mathfrak C))=0$ for every $i$. Proposition \ref{prop:mob} allows us to choose a subgroup $\Gamma \subset \Lambda$ of finite index such that $x^\gamma-1$ annihilates all $Q^i(\mathfrak C)$ for all $\gamma \in \Gamma$. With this $\Gamma$, we consider the \v{C}ech complex as discussed around \eqref{eq:Cech_setup}.
\subsubsection{Charges as \v Cech cocyles}
\begin{prop} \label{prop:charges_Cech}
We have $Q^i(\mathfrak C) \cong \CH^{i+1}(\mathbf t, P/L)$ for $0 \leq i \leq D-2$.
\end{prop}
\begin{proof}
We can continue the quotient map $P \to P/L$ to a quasi-free resolution $P_\bullet \to P/L$ with $P_0 = P$. Applying $(-)^*$ yields a complex
\begin{equation}
    0 \to L \to P \to P_1^* \to P_2^* \to \dots,
\end{equation}
where we used isomorphisms $(P/L)^* \cong L$ and $P^* \cong P$. From this we have also a cochain complex $K^{\bullet}$ with $K^0 = P/L$, $K^i = P_i^*$ for $i>0$:
\begin{equation}
  K^{\bullet} \  : \   0 \to P/L \to P_1^* \to P_2^* \to \dots
\end{equation}
Its cohomology is trivial in degree zero and $\Ext^\bullet_R(\overline{P/L},R)$ elsewhere. Next, we form a double complex $\Cech^{\bullet}(\mathbf t,K^{\bullet})$, with the following properties:
\begin{itemize}
    \item $\Cech^0(\mathbf t, K^\bullet) \cong K^\bullet$ has cohomology described above. If $p>0$, the complex $\Cech^p(\mathbf t, K^\bullet)$ is exact because $\Cech^p(\mathbf t, -) = \Cech^p(\mathbf t, R) \otimes_R -$ is an exact functor annihilating the cohomology of $K^\bullet$.
    \item $\Cech^\bullet(\mathbf t, K^0)$ has cohomology $\CH^\bullet(\mathbf t,P/L)$. If $q>0$, the complex $\Cech^\bullet(\mathbf t, K^q)$ has nonzero cohomology only in degree $D$, by Proposition \ref{prop:loc_ch_R}.
\end{itemize}
The isomorphism is established either by a~diagram chase or using the double complex spectral sequence. The former approach is essentially elementary and we sketch it below.

We let $d$ be the differential induced from $K^{\bullet}$ and $\delta$ the \v{C}ech differential. Let $i \in \{ 1, \dots D-1 \}$ and consider $q \in \Ext^{i}_R(\overline{P/L},R)$ represented by an element $q^{(0)} \in K^i$ annihilated by $d$. Then also $\delta q^{(0)}$ is annihilated by $d$, so by exactness of $\Cech^1(\mathbf t , K^\bullet)$ there exists $q^{(1)} \in \Cech^{1}(\mathbf t, K^{i-1})$ such that $d q^{(1)} = \delta q^{(0)}$. Hence $\delta q^{(1)}$ is annihilated by $d$. If $i=1$, this implies that $\delta q^{(1)} = 0$ because $d : \Cech^2(\mathbf t, K^0) \to \Cech^2 (\mathbf t, K^1)$ is injective. If $i>1$, we conclude that there exists $q^{(2)} \in \Cech^2(\mathbf t, K^{i-2})$ such that $d q^{(2)} = \delta q^{(1)}$. Continuing like this inductively we obtain a sequence of elements $q^{(j)} \in \Cech^j(\mathbf t, K^{i-j})$, $0 \leq j \leq i$, such that
\begin{equation}
d q^{(0)}=0, \qquad d q^{(j)} = \delta q^{(j-1)} \quad \text{for } j \neq 0, \qquad \delta q^{(i)}=0.
\end{equation}
The \v{C}ech cohomology class of $q^{(i)}$ is declared to be the image of $q$ in $\CH^i(\mathbf t, P/L)$. With similar reasoning one checks that this cohomology class does not depend on arbitrary choices in the construction of $q^{(i)}$. Thus a well-defined homomorphism $h : \Ext^i_R(\overline{P/L},R) \to \CH^i(\mathbf t,P/L)$ is obtained. Performing the same steps reversed yields a homomorphism in the opposite direction, easily seen to be an inverse of $h$.
\end{proof}

\begin{rmk}
If we assume that $x^{\gamma}-1$ annihilates $Q^i(\mathfrak C)$ for every $i \leq d$ for some $0 \leq d \leq D-2$, we may still obtain $Q^i(\mathfrak C) \cong \CH^{i+1}(\mathbf t, P/L)$ for $0 \leq i \leq d$. The proof of Proposition \ref{prop:charges_Cech} goes through with essentially no modifications. Moreover, even with no restrictions on $\dim(Q^i(\mathfrak C))$ we may construct a homomorphism $\CH^{i}(\mathbf t, P/L) \to \Ext^{i}_R(\overline{P/L},R)$ for $1 \leq i \leq D-1$. If $\dim (\Ext^{i}_R(\overline{P/L},R)) \neq 0$, this homomorphism cannot be surjective.
\end{rmk}
\subsubsection{Charges as topological excitations}
Next we provide a concrete interpretation of our charge modules $Q^i(\mathfrak C)$ (reinterpreted as \v{C}ech cocycles by Proposition \ref{prop:charges_Cech}) in terms of operators and physical excitations.

\begin{defn}
We define $\widehat P = P \otimes_R R^{\#}$. Recall that $P \cong P_0[\Lambda]$ for some finite abelian group $\Lambda$, so $\widehat P \cong P_0 [[\Lambda]]$. We will sometimes multiply elements of $R^{\#}$ and $\widehat P$. Such product is well-defined under a condition analogous to the one discussed in Definition \ref{def:formal_series}. Symplectic form on $P$ extends to a pairing between $P$ and $\widehat P$ valued in $R^{\#}$. Under suitable conditions one may also pair two elements of $\widehat P$.
\end{defn}

Elements of $\widehat P$ describe products of Pauli operators (up to phase) with possibly infinite spatial support. Such expressions do not necessarily define bona fide operators on a Hilbert space, but they make sense as automorphisms of the algebra of local operators. Hence they may be applied to states, in general yielding a state in a different superselection sector. The extended symplectic forms captures their ``commutation rules'' with local Pauli operators. 

\begin{defn}
Let $s=(s_1, \dots, s_D)$ be a tuple of elements of the multiplicative group $\{ \pm \}$. We think of $s$ as a label of an orthant in $\Gamma \cong \ZZ^D$. For every $s$ we define an embedding of $P_{t_1 \dots t_D} $ (and hence also of every $P_{t_{i_0 \dots i_p}}$ for a sequence $1 \leq i_0 < \dots < i_p \leq D$, since $P$ is torsion-free) in $\widehat P$ as follows:
\begin{equation}
    \frac{p}{t_1^{k_1} \cdots t_D^{k_D}} \mapsto p \prod_{i=1}^D \left( \frac{1}{t_i} \right)_{s_i}^{k_i}.
\end{equation}
If $ \pi $ is an element of $P_{t_{1} \dots t_D}$, we denote the element of $\widehat P$ obtained this way by $\pi^s$, to emphasize dependence on $s$. 
\end{defn}

Consider a cocycle $\varphi \in \Cech^p(\mathbf t, P/L)$. We lift $\varphi$ to a cochain $\widetilde \varphi \in \Cech^p(\mathbf t, P)$. Then $\sigma = \delta \widetilde \varphi \in \Cech^{p+1}(\mathbf t, L)$ is a cocycle. Note that the map taking the cohomology class of $\varphi$ to the cohomology class of $\sigma$ is the connecting homomorphism in the long exact sequence of \v{C}ech cohomology. Consider images in $\widehat P$ of components of $\widetilde \varphi$ and $\sigma$. Two observations are in order. Firstly, $\widetilde \varphi_{i_1 \dots i_p}^s$ describes an infinite Pauli operator whose support is extended only in directions $i_1 \dots i_p$, and moreover is contained in a shifted orthant specified by $s$. Secondly, each $\sigma_{i_0 \dots i_p}^s$ is $\omega$-orthogonal to $L$. Hence we have an identity
\begin{equation}
\sum_{j=0}^p (-1)^j \omega \left. \left( \cdot , \widetilde \varphi_{i_0 \dots i_{j-1} i_{j+1} \dots i_p}^s \right) \right|_{L} = 0 \qquad \text{in } \overline L^{\#}.
\end{equation}
Let us rewrite this as
\begin{equation}
    \omega \left. \left( \cdot , \widetilde \varphi_{i_1 \dots i_p}^s \right) \right|_{L} = \sum_{j=1}^p (-1)^j \omega \left. \left( \cdot , \widetilde \varphi_{i_0 \dots i_{j-1} i_{j+1} \dots i_p}^s \right) \right|_{L}
\end{equation}
By comparing supports of the two sides of this equation we can see that action of $\widetilde \varphi_{i_1 \dots i_p}^s$ creates an excitation (violation of the stabilizer condition) which is supported on a thickened boundary of the support of $\widetilde \varphi^s_{i_1 \dots i_p}$. Hence $\widetilde \varphi^s_{i_1 \dots i_p}$ represents a~$p$-dimensional extended operator which creates an excitation on the $(p-1)$-dimensional boundary of its support. This excitation does not depend on the lift of the cocycle $\varphi$ to $\widetilde \varphi$.

Next, let us suppose that $\varphi$ represents the trivial cohomology class. That is, we have $\varphi = \delta \psi$ for some $\psi \in \Cech^{p-1}(\mathbf t, P/L)$. We lift $\psi$ to a cochain $\widetilde \psi$ valued in $P$ and choose $\widetilde \varphi = \delta \widetilde \psi$. Then
\begin{equation}
    \widetilde \varphi_{i_1 \dots i_p}^s = \sum_{j=1}^{p} (-1)^{j-1} \widetilde \psi_{i_1 \dots i_{j-1} i_{j+1} \dots i_p}^s,
\end{equation}
which shows that the $(p-1)$-dimensional excitation created by $\widetilde \varphi^s_{i_1 \dots i_p}$ can be created by operators $\widetilde \psi_{i_1 \dots i_{j-1} i_{j+1} \dots i_p}^s$, each of which is extended in only $p-1$ (rather than $p$) directions.

Note that even though an excitation corresponding to a $p$-cocycle $\varphi$ is created by an operator with $p$-dimensional support, it can be shifted by an element of $\Gamma$ by the action of a $(p-1)$-dimensional operator. Indeed, $x^{\gamma}-1$ annihilates cohomology, so $(x^\gamma-1) \varphi$ is a coboundary. The result follows from the discussion of the previous paragraph. 

Summarizing, an element of $Q^p(\mathfrak C) \cong \CH^{p+1}(\mathbf t, P/L)$ gives rise to an excitation extended in $p$ dimensions, determined modulo excitations created by $p$-dimensional operators.

\subsubsection{Charges as higher form symmetries}

Now let $\varphi \in \Cech^p(\mathbf t, P/L)$ be a cocycle. We consider the expression
\begin{equation}
    \widetilde \varphi^{\Res}_{i_1 \dots i_p} = \sum_{s_{i_1}, \dots, s_{i_p} \in \{\pm\}} s_{i_1} \cdots s_{i_p} \, \widetilde \varphi^s_{i_1 \dots i_p}.
    \label{eq:Res_operator}
\end{equation}
This makes sense because $\widetilde \varphi^s_{i_1 \dots i_p}$ does not depend on $s_j$ for $j \not \in \{ i_1 , \dots ,i_p \}$. $\widetilde \varphi^{\Res}_{i_1 \dots i_p}$ is a $p$-dimensional extended operator. By the earlier discussion, the excitation it creates is supported in the union of a finite collection of subsets infinitely extended in at most $p-1$ directions. On the other hand, there exists some $k$ such that each $t_{i_j}^k$ annihilates it. One checks that a nonzero element with such property must be infinitely extended in all $p$ directions. We obtain the conclusion that $\left. \omega (\cdot , \widetilde \varphi^{\Res}_{i_1 \dots i_p}) \right|_{L}=0$, i.e. $\widetilde \varphi^{\Res}_{i_1 \dots i_p}$ preserves the state defined by the stabilizer condition.

Since the cochain $\widetilde \varphi$ allows to construct a symmetry $\widetilde \varphi^{\Res}_{i_1 \dots i_p}$ of the ground state for every coordinate $p$-plane (labeled by $i_1 < \dots <i_p$), it defines a~$(D-p)$-form symmetry of $\mathfrak C$. Let us now investigate to what extent this $(D-p)$-form symmetry is uniquely determined by the cohomology class of $\varphi$.

Firstly, let us fix the cocycle $\varphi$ and ask for the dependence on the choice of the lift $\widetilde \varphi$. For two different lifts $\widetilde \varphi, \widetilde \varphi'$, the difference $\widetilde \varphi_{i_1 \dots i_p}'^{\Res}-\widetilde \varphi_{i_1 \dots i_p}^{\Res}$ is an infinite sum of elements of $L$, i.e. it represents a product of local operators separately preserving the ground state. A $p$-dimensional ($p \geq 1$) operator of this form should be regarded as a trivial $(D-p)$-form symmetry.

To understand the dependence on the cocycle $\varphi$ representing a given cohomology class, let us suppose that $\varphi = \delta \psi$. We lift $\psi$ and choose $\widetilde \varphi = \delta \widetilde \psi$. With this choice expression \eqref{eq:Res_operator} vanishes on the nose.

Summarizing, we have argued that the definition \eqref{eq:Res_operator} defines a $(D-p)$-form symmetry of $\mathfrak C$, which depends only on the cohomology class of $\varphi$. This means that we have an alternative interpretation of $Q^p(\mathfrak C)$ as a group of $(D-p-1)$-form symmetries of $\mathfrak C$ (possibly nontrivially acted upon by $\Lambda$).

\subsection{Braiding}

\begin{defn}
Let $\varphi \in \CH^{p}(\mathbf t, P/L)$, $\psi \in \CH^{q}(\mathbf t, P/L)$. The cup product defined in the Appendix \ref{sec:Cech} yields an element
\begin{equation}
    \overline \varphi \smile \delta \psi \in \CH^{p+q}(\mathbf t, \overline{P/L} \otimes_R L),
\end{equation}
where $\delta$ is the connecting homomorphism $\CH^{q}(\mathbf t, P/L) \to \CH^{q+1}(\mathbf t, L)$ in a long exact sequence. Using the map (with a slight abuse of notation) in \v{C}ech cohomology induced by the symplectic pairing  $\omega: \overline{P/L} \otimes_R L \to R$ we obtain a class 
\begin{equation}
    \omega(\overline \varphi \smile \delta \psi) \in  \CH^{p+q}(\mathbf t, R).
\end{equation}
This class is trivial if $p+q \neq D$, by Proposition \ref{prop:loc_ch_R}. Let us suppose that $p+q = D$. Then we may define
\begin{equation}
    \Omega (\varphi, \psi) = \Res  \left( \omega(\overline \varphi \smile \delta \psi) \right) \in \Gamma_{\mathfrak a}(R^{\#}).
\end{equation}
\end{defn}

\begin{prop}
Let $\varphi \in \CH^{p}(\mathbf t, P/L)$, $\psi \in \CH^{D-p}(\mathbf t, P/L)$. We have:
\begin{enumerate}
    \item Graded skew-symmetry: $\Omega(\varphi, \psi)  = - (-1)^{p(D-p)} \overline{\Omega(\psi,\varphi)}$.
    \item Translation covariance: $ \Omega(\varphi, r\psi)  =  \Omega(\overline r \varphi,\psi) = r \Omega(\varphi, \psi)$.
    \item Commutation rule of operators introduced in Subsection \ref{sec:operators}:
    \begin{equation}
    \Omega(\varphi, \psi) = \Res ( \omega(\widetilde \varphi_{1 \dots p}, \widetilde \psi_{p+1 \dots D})) =\omega( \widetilde \varphi_{1 \dots p}^{\Res}, \widetilde \psi_{p+1 \dots D}^{\Res}). 
    \label{eq:Omega_operators}
    \end{equation}
\end{enumerate}
\end{prop}
\begin{proof}
1. follows from the graded commutativity and graded Leibniz rule of the cup product and antipode skew-symmetry of $\omega$. 2. is obvious.

3. From the relevant definitions we have
\begin{equation}
    \Omega(\varphi, \psi) = \Res \sum_{j=0}^{D-p} (-1)^j \omega(\widetilde \varphi_{1 \dots p}, \widetilde \psi_{p \dots p+j-1, p+j+1 \dots D}).
    \label{eq:Omega_intermediate}
\end{equation}
Let $0 < j \leq D-p$. The $j$-th term on the right hand side of \eqref{eq:Omega_intermediate} is the residue of an element of $R_{t_1 \dots t_{p-j-1} t_{p-j+1} \dots D}$, so it vanishes. The $0$-th term is equal to the right hand side of \eqref{eq:Omega_operators}.
\end{proof}

We propose to interpret the scalar part of $\Omega$ as a higher dimensional version of braiding. Thus $\Omega(\varphi, \psi)$ encodes braiding of excitations described by $\varphi, \psi$ as well as their translates. We will see later that for $D=2$ our proposal reduces to known expressions, providing evidence for our interpretation. 

Recall that we have a decomposition $\CH^{p+1}(\mathbf t, P/L) = \bigoplus_{\mathfrak m} \Gamma_{\mathfrak m} \CH^{p+1}(\mathbf t, P/L)$, where $\mathfrak m$ are maximal ideals of $R$ containing $\mathfrak a$. Its summands are charges characterized by specific behavior under translations, so we interpret $\mathfrak m$ as momentum ``quantum numbers''. Note that for every $\mathfrak m$, the ideal $\overline {\mathfrak m}$ obtained by acting with the antipode also contains $\mathfrak a$, as $\overline{\mathfrak a} = \mathfrak a$. We think of $\overline{\mathfrak m}$ as momentum opposite to $\mathfrak m$. The following Proposition shows that two charges with fixed momentum may braid nontrivially only if their momenta are opposite.

\begin{prop} \label{prop:braiding_orthogonality}
Suppose that $\varphi \in \Gamma_{\mathfrak m} \CH^{p+1}(\mathbf t, P/L), \psi \in \Gamma_{\mathfrak m'} \CH^{D-p-1}(\mathbf t, P/L)$. If $\overline{\mathfrak m} \neq \mathfrak m'$, then $\Omega(\varphi , \psi) = 0$.
\end{prop}
\begin{proof}
For some $j$, $\varphi$ is annihilated by $\mathfrak m^j$ and $\psi$ by $\mathfrak m'^j$. Therefore $\Omega(\varphi , \psi)$ is annihilated by $\overline{\mathfrak m}^j + \mathfrak m'^j$. If $\overline{\mathfrak m} \neq \mathfrak m'$, this sum is $R$.
\end{proof}

Decomposition of $\CH^{p+1}(\mathbf t, P/L)$ into $\mathfrak m$-torsion parts is not invariant to coarse-graining. In fact, after sufficient coarse-graining we can assure that $\mathfrak a$ contains all $x^\lambda-1$. Then, for $n$ being a prime power, $\mathfrak a$ is contained in only one maximal ideal. Decomposition into $\mathfrak m$-torsion parts (and more generally, the module structure on $Q^i(\mathfrak C)$) is an invariant protected by the translation symmetry and hence in principle can be used to distinguish SET phases with the same topological order. 

\subsection{Braiding and spin in 2D}

We will now specialize to 2D Lagrangian codes. The assumption $\dim(Q)=0$ is automatically satisfied, as stated in Corollary \ref{cor:2d_DIM}. Hence we have well-defined braiding. Expression \eqref{eq:Omega_operators} agrees with the standard braiding formula as a commutator of two orthogonal string operators. Let us explain this in more detail.

Consider a Lagrangian $\mathfrak C=(\ZZ^2, P, L)$. We have $R = \ZZ_n [x^\pm_1,x^\pm_2]$. There exists some $l>0$ such that $t_i = 1 - x_i^l \in \ann (Q(\mathfrak C))$. Therefore we have the following commutative diagram with exact rows
\begin{figure}[h]
\centering
\begin{tikzcd}[cells={nodes={minimum height=2em}}]
0\arrow[r] & P/L \arrow[r, "\iota_0"] \arrow[d] &  L^*  \arrow[d, "\delta"] \arrow[r] & Q(\mathfrak C) \arrow[r] & 0 \\
0 \arrow[r] & \Cech^1(\mathbf t,P/L) \arrow[r,"\iota_1"] &\Cech^1(\mathbf t,L^*) \arrow[r] & 0
\end{tikzcd}
\end{figure} \\
For any $e\in L^*$, we have $\delta e=(e,e)= \iota_1 (\frac{p_1}{x_1^l-1}, \frac{p_2}{x_2^l-1})$ with $p_i=(x_i^l-1) e\in P/L$. One may check that $e\longmapsto (\frac{p_1}{x_1^l-1}, \frac{p_2}{x_2^l-1})$ defines an isomorphism between $L^*$ and $1$-cocycles, with elements of $P/L$ mapped onto coboundaries. In particular this map induces an isomorphism $L^* / (P/L) \to \CH^1(\mathbf t, P/L)$. A lift of $p_i$ to $P$ represents a Pauli operator which moves the excitation $e$ by $l$ units in the $i$-th direction. For this reason, $p_i$ is sometimes called an $i$-mover.

Let $e_1,e_2 \in L^*$ be two excitations and let $p_i(e_j)$ be their movers. We can then form arbitrarily long string operators
\begin{equation}
    (x_i^{-cl} + \cdots + 1 + \cdots x_i^{cl}) p_i(e_j)
\end{equation}
which transport (displaced) excitations described by $e_j$ by $(2c+1)l$ units of length. Braiding may be related \cite{ haah2021classification, levin2003fermions} to the commutator phase
\begin{equation}
    \omega((x_1^{-cl}+\cdots+x_1^{cl})p_1(e_1), (x_2^{-cl}+\cdots  +x_2^{cl})p_2(e_2))_0
\end{equation}
with sufficiently large $c$. This expression is asymptotically independent of $c$ because the two strings operators cross at most along a finite set. Taking $c$ to infinity, this expression matches the scalar part of \eqref{eq:Omega_operators} with
\begin{equation}
    \varphi= \left( \frac{p_1(e_1)}{x_1^l-1}, \frac{p_2(e_1)}{x_2^l-1} \right), \qquad \psi=\left( \frac{p_1(e_2)}{x_1^l-1}, \frac{p_2(e_2)}{x_2^l-1} \right).
\end{equation}
We remark that it is also equal to the evaluation of the Laurent polynomial $\omega(p_1(e_1),p_2(e_2))$ at $x_1=x_2=1$.

One can also define the topological spin function
\begin{align}
        \theta(e)
    =&\omega((x_1^{-cl}+\cdots+x_1^{-l})p_1(e), (x_2^{-cl}+\cdots  +x_2^{-l})p_2(e))_0 \notag \\
    -& \omega( (x_2^{-cl}+\cdots  +x_2^{-l})p_2(e),(1+x_1^{l}+\cdots+x_1^{cl})p_1(e))_0 \notag\\
    -& \omega((1+x_1^{l}+\cdots+x_1^{cl})p_1(e), (x_1^{-cl}+\cdots+x_1^{-l})p_1(e))_0 \label{eq:top_spin}
\end{align}
with sufficiently large $c$ (the right hand side, as a function of $c$, is eventually constant). It is a quadratic refinement of the braiding pairing: 
\begin{equation}
  \Omega(e,f)_0 = \theta(e+f)-\theta(e)-\theta(f) \text{ and } \theta(ke)= k^2 \theta(e) \text{ for } k \in \mathbb Z_n.
\end{equation}
Formula \eqref{eq:top_spin} appeared first in \cite{haah2021classification}, where the case of prime-dimensional qudits was studied.

\section{Examples}

In this section we discuss examples with concrete codes. They serve several purposes. Firstly, they show that invariants we proposed are nontrivial, calculable and yield what is expected on physical grounds in models which are already well understood. Secondly, they support our physical interpretation of mathematical objects and the conjecture that braiding is non-degenerate. Finally, the last example illustrates certain technical complication that does not arise for codes with
prime-dimensional qudits.

In examples presented below we take $P$ to be a free module $R^{2t}$ with the symplectic form
\begin{equation}
    \omega \left( \begin{pmatrix} a \\ b \end{pmatrix}, \begin{pmatrix} a' \\ b' \end{pmatrix} \right) = \begin{pmatrix} a^\dagger & b^\dagger \end{pmatrix} \underbrace{\begin{pmatrix} 0 & -1 \\ 1 & 0 \end{pmatrix}}_{\text{denote } \lambda} \begin{pmatrix} a' \\ b' \end{pmatrix},
\end{equation}
where $a,a',b,b' \in R^t$ and $\dagger$ denotes transposition composed with antipode. Following \cite{haah2013commuting}, we represent $L$ as the image of a homomorphism $\sigma : R^s \to R^{2t}$, described by a~$2t \times s$ matrix with entries in $R$.

We will also work with cocycles in $\Cech^\bullet(\mathbf t, P/L)$. In calculations it is convenient to identify them with cochains in $\Cech^\bullet(\mathbf t,P)$ which are closed modulo $\Cech^\bullet(\mathbf t, L)$, with two cochains identified if they differ by a cochain in $\Cech^\bullet(\mathbf t, L)$. 

\subsection{3D \texorpdfstring{$\ZZ_n$}{Zn}-toric code}

We take $\Lambda = \ZZ^3$ and denote generators of $R$ corresponding to three basis vectors by $x,y,z$, so that $R$ is a Laurent polynomial ring in three variables $x,y,z$. 3D toric code is defined by $P = R^6$, $L = \im (\sigma)$ with
\begin{equation}
\sigma= \begin{pmatrix}
1-\bar{x} & 0 & 0 & 0\\
1-\bar{y} & 0 & 0 & 0\\
1-\bar{z} & 0 & 0 & 0\\
0 & 0 & z-1 & y-1\\
0 & z-1 & 0 & 1-x\\
0 & 1-y & 1-x & 0
\end{pmatrix}.
\end{equation}
We have the following free resolution of $P/L$
\begin{equation}
0\rightarrow R\xrightarrow{\tau} R^4 \xrightarrow{\sigma} R^{6} \rightarrow P/L \to 0, \qquad \tau=\begin{pmatrix}0\\
x-1\\ 1-y\\ z-1
\end{pmatrix}.
\end{equation}
Erasing $P/L$ and applying $(-)^*$ we obtain
\begin{equation}
0 \to P \xrightarrow{\sigma^\dagger \lambda}  R^4\xrightarrow{\tau^\dagger} R\rightarrow 0.
\end{equation}
Here matrix $\epsilon = \sigma^\dagger \lambda$ (rather than $\sigma^\dagger$) is present because the canonical isomorphism $P \to P^*$ is given by $\lambda$ if both $P$ and $P^*$ are identified with $R^6$. From this resolution we easily get
\begin{align}
    & \Ext^1_R(\overline{P/L},R) \cong \ZZ_n, \quad && \text{generated by the class of } \begin{pmatrix} 1 & 0 & 0 & 0 \end{pmatrix}^{\mathrm{T}} \in R^4, \nonumber \\
    & \Ext^2_R(\overline{P/L},R) \cong \ZZ_n, \quad && \text{generated by the class of } 1 \in R.
\end{align}
Both Ext modules are annihilated by $x-1,y-1,z-1$.

Let us show how \v{C}ech cochains can be obtained from classes found above. The procedure below follows from proof of Proposition 44. In the construction of the \v{C}ech complex we may take $(x_1,x_2,x_3)= (x,y,z)$. Recall that we defined $t_i = 1 - x_i$. Now consider $\begin{pmatrix} 1 & 0 & 0 & 0 \end{pmatrix}^{\mathrm T} \in R^4$. Applying the \v{C}ech differential gives
\begin{align}
 R_{t_1}^4 \oplus R_{t_2}^4 \oplus R_{t_3}^4 & \ni  \left( \begin{pmatrix} 1 \\ 0 \\ 0 \\ 0 \end{pmatrix}, \begin{pmatrix} 1 \\ 0 \\ 0 \\ 0 \end{pmatrix}, \begin{pmatrix} 1 \\ 0 \\ 0 \\ 0 \end{pmatrix} \right) \\
  & = \left( \epsilon \begin{pmatrix} 0 \\ 0 \\ 0 \\ -t_1^{-1} \\ 0 \\ 0 \end{pmatrix}, \epsilon \begin{pmatrix} 0 \\ 0 \\ 0 \\ 0 \\ -t_2^{-1} \\ 0 \end{pmatrix}, \epsilon \begin{pmatrix} 0 \\ 0 \\ 0 \\ 0 \\ 0 \\ -t_3^{-1} \end{pmatrix} \right). \nonumber
\end{align}
The final expression is the image through $\epsilon$ of a certain element of $\Cech^1(\mathbf t, P)$. Let us call this cochain $\varphi$. By construction, it is closed modulo $L$. Let us show how this can be checked by an explicit computation:
\begin{align}
& \Cech^2(\mathbf t, R^6) =R^6_{t_2 t_3} \oplus R^6_{t_1 t_3} \oplus R^6_{t_1 t_2}   \ni \delta \varphi \nonumber  \\
& = \left( \begin{pmatrix} 0 \\ 0 \\ 0 \\ 0 \\ t_2^{-1} \\ -t_3^{-1} \end{pmatrix}, \begin{pmatrix} 0 \\ 0 \\ 0 \\ t_1^{-1} \\ 0 \\ -t_3^{-1} \end{pmatrix}, \begin{pmatrix} 0 \\ 0 \\ 0 \\ t_1^{-1} \\ -t_2^{-1} \\ 0 \end{pmatrix}  \right)    \\
& = \left( \sigma \begin{pmatrix} 0 \\ - t_2^{-1} t_3^{-1} \\ 0 \\ 0 \end{pmatrix}, \sigma \begin{pmatrix} 0 \\ 0 \\ -t_1^{-1} t_3^{-1} \\ 0 \end{pmatrix}, \sigma \begin{pmatrix} 0 \\ 0 \\ 0 \\ - t_1^{-1} t_2^{-1} \end{pmatrix} \right). \nonumber
\end{align}

One can go through a similar procedure with the element generating $\Ext^2$. Let us record the final result:
\begin{equation}
\Cech^2(\mathbf t, R^6) \ni \psi = \left( \begin{pmatrix} \overline t_2^{-1} \overline t_3^{-1} \\ 0 \\0 \\ 0 \\ 0 \\ 0 \end{pmatrix}, \begin{pmatrix}  0 \\ - \overline t_1^{-1} \overline t_3^{-1} \\ 0 \\ 0 \\ 0 \\ 0  \end{pmatrix}, \begin{pmatrix}  0 \\ 0 \\ \overline t_1^{-1} \overline t_2^{-1} \\ 0 \\ 0 \\ 0 \end{pmatrix} \right).
\end{equation}
Having these formulas in hand we evaluate
\begin{equation}
    \Omega(\varphi, \psi )_0 = 1 \in \ZZ_n.
\end{equation}
Hence braiding is a non-degenerate pairing in this example.

It is well known that toric code is closely related to $\ZZ_n$ gauge theory. With this interpretation, line operators corresponding to $\varphi$ are Wilson lines. They create electric excitations at their endpoints. Cocycle $\psi$ corresponds to electric flux (surface) operators, which create magnetic field on the boundary. Braiding between the two excitations is an Aharonov-Bohm type phase. We remark also that the relation between generators of $L$, described by the map~$\tau$, corresponds to Bianchi identity. 

\subsection{4D \texorpdfstring{$\ZZ_n$}{Zn}-toric code}

In a 4D version of the $\ZZ_n$ toric code we have $P = R^{8}$. We let $x_1, \dots, x_4$ be four variables corresponding to generators of $\ZZ^4$ and denote basis vectors of $P$ by $e_1, \dots , e_4, a_1, \dots , a_4$. Consider the free module $R^7$ with basis $\{ g \} \cup \{ f_{ij} \}_{1 \leq i < j \leq 4}$. We define $L = \im(\sigma)$, where $\sigma : R^7 \to P$ is given by
\begin{equation}
    \sigma(g) = \sum_{i=1}^4 (1 - \overline x_i) e_i, \qquad \sigma(f_{ij}) = - (x_i-1) a_j + (x_j-1) a_i.
\end{equation}
Elements $\sigma(g), \sigma(f_{ij})$ generate $L$. To continue $\sigma$ to a resolution of $P/L$, we need to describe relations between generators. Consider the free module $R^4$ with basis $\{ b_{ijk} \}_{1 \leq i < j < k \leq 4}$. Define $\tau_1 : R^4 \to R^7$ by
\begin{equation}
    \tau_1(b_{ijk}) = (x_i-1) f_{jk} - (x_j-1) f_{ik} + (x_k-1) f_{ij}.
\end{equation}
Then $\im(\tau_1) = \ker(\sigma)$, but we still have to take care of relations between relations. Let $\tau_2 : R \to R^4$ be given by
\begin{equation}
    \tau_2 (1) = (x_1-1) b_{234} - (x_2-1) b_{134} + (x_3-1) b_{124} - (x_4-1) b_{123}.
\end{equation}
We have constructed a free resolution
\begin{equation}
    0 \to R \xrightarrow{\tau_2} R^4 \xrightarrow{\tau_1} R^7 \xrightarrow{\sigma} R^8 \to P/L \to 0
\end{equation}
Proceeding as in the 3D case we found
\begin{equation}
    Q^0 \cong \ZZ_n, \qquad Q^1 = 0, \qquad Q^2 \cong \ZZ_n,
\end{equation}
all annihilated by $x_i-1$. After some tedious calculations we found also the \v Cech cochains $\varphi \in \Cech^1(\mathbf t, P)$ and $\psi \in \Cech^3(\mathbf t, P)$ corresponding to generators of $Q^0$ and $Q^2$:
\begin{equation}
\varphi_i = \frac{a_i}{x_i-1}, \qquad \psi_{i^c} = \frac{(-1)^i e_i}{ \prod_{j \neq i} (\overline x_j -1)},
\end{equation}
where $i^c$ denotes the triple of indices complementary to $i$. Given these expressions it is easy to check that
\begin{equation}
    \Omega(\varphi, \psi)_0 = 1.
\end{equation}
Again, braiding is non-degenerate.

\subsection{4D \texorpdfstring{$\ZZ_n$}{Zn} 2-form toric code}

By a 2-form version of the toric code we mean a code in which degrees of freedom are assigned to lattice plaquettes. Starting from dimension $4$ such code is neither trivial nor equivalent to the standard (`1-form') toric code. Module $P \cong R^{12}$ has basis $\{ e_{ij}, a_{ij} \}_{1 \leq i < j \leq 4}$, with nontrivial symplectic pairings of the form $\omega(e_{ij},a_{ij})=1$. Consider the free module $R^8$ with basis $\{ g_i , f_{i^c} \}_{i=1}^4$. We define $L = \im(\sigma)$, where $\sigma : R^8 \to P$ is given by
\begin{align}
    \sigma(g_i) & = - \sum_{j<i } (1 - \overline x_j) e_{ji} + \sum_{j > i} (1 - \overline x_j) e_{ij}, \\
    \sigma(f_{ijk}) &= - (x_i-1) a_{jk} + (x_j-1) a_{ik} - (x_k-1) a_{ij}. \nonumber
\end{align}
$\ker(\sigma)$ coincides with the image of $\tau : R^2 \to R^8$ such that
\begin{equation}
    \tau \begin{pmatrix} 1 \\ 0 \end{pmatrix} = \sum_i (1 - \overline x_i) g_i, \qquad \tau \begin{pmatrix} 0 \\ 1 \end{pmatrix} = \sum_i (-1)^i (x_i-1) f_{i^c}.
\end{equation}
This defines a free resolution
\begin{equation}
    0 \to R^2 \xrightarrow{\tau} R^8 \xrightarrow{\sigma} R^{12} \to P/L \to 0,
\end{equation}
from which we derive
\begin{equation}
Q^0 = 0, \qquad Q^1 \cong \ZZ_n \oplus \ZZ_n, \qquad Q^2 =0,
\end{equation}
with $Q^1$ annihilated by all $x_i-1$. Two \v Cech cochains corresponding to generators of $Q^1$ take the form
\begin{equation}
    \varphi_{ij} = \frac{a_{ij}}{(x_i-1)(x_j-1)}, \qquad \psi_{ij} = \frac{(-1)^{i+j} e_{ij^c}}{(\overline{x_i}-1)(\overline{x_j}-1)},
\end{equation}
where $ij^c$ is the pair of indices complementary to $ij$. We find
\begin{equation}
    \Omega(\varphi, \psi)_0 = 1,
\end{equation}
so braiding is non-degenerate.

\subsection{\texorpdfstring{$\ZZ_n$}{Zn} Ising model}

For the Ising model in zero magnetic field we have $P = R^2$ and $L = \im(\sigma)$, where
\begin{equation}
    \sigma = \begin{pmatrix} x_1 -1 & \cdots & x_D-1 \\ 0 & \cdots & 0 \end{pmatrix},
\end{equation}
where $D \geq 1$ is arbitrary. We see that $\begin{pmatrix} 1 \\ 0 \end{pmatrix} \in L^{\omega \omega} \setminus L$ and $L^{\omega \omega} / L \cong \ZZ_n$, in accord with the interpretation of $L^{\omega \omega}/ L$ in terms of order parameters for spontaneously broken symmetries. Next, we note that $P/L \cong R \oplus R/\mathfrak a$, where $\mathfrak a=(x_1-1, \cdots, x_D-1) $. Hence for every $i>0$ we have
\begin{equation}
    \Ext^i_R(\overline{P/L},R) \cong \Ext^i_R(R/\mathfrak a,R).
\end{equation}
As elements $x_i-1$ form a regular sequence in $R$, this Ext vanishes for $i \neq D$ and $\Ext^D_R(\overline{P/L},R) \cong R / \mathfrak a$. Therefore the only nonzero $Q^i$ is $Q^{D-1} \cong \ZZ_n$. This is consistent with the interpretation of $Q^i$ in terms of $i$-dimensional excitations: the Ising model features domain walls, which are objects of spatial codimension~$1$. However, our formalism does not provide a systematic construction of this domain wall (Ising model is not a Lagrangian code). Let us also remark that we expect that there exists a generalization of braiding that allows to pair $Q^{D-1}$ with $L^{\omega \omega}/L$. Physically such pairing should describe how the value of order parameter changes as the domain wall is crossed.

\subsection{\texorpdfstring{$\ZZ_n$}{Zn} toric code on a cylinder}

Consider the 2D cylinder geometry $\Lambda = \ZZ_L \times \ZZ$. Thus $R = \ZZ_n[x,y^\pm]/(x^L-1)$. We let $P = R^4$ and $L = \im(\sigma)$, where 
\begin{equation}
    \sigma = 
    \begin{pmatrix} 
    1 - \overline x & 0 \\
    1 - \overline y & 0 \\ 
    0 & y-1 \\ 
    0 & x-1 
    \end{pmatrix}.
\end{equation}
Let us put $W_x = \sum_{j=0}^{L-1} x^{j-1} \in R$. Note that $(x-1) W_x = 0$, so
\begin{equation}
(y-1) \begin{pmatrix} 0 \\ 0 \\ W_x \\ 0 \end{pmatrix} =  W_x \begin{pmatrix} 0 \\ 0 \\ y-1 \\ x-1 \end{pmatrix} \in L.
\end{equation}
Since $y-1$ is a regular element, it follows that $\begin{pmatrix} 0 & 0 & W_x & 0 \end{pmatrix}^{\mathrm T} \in L^{\omega \omega}$. Similar calculation shows that $\begin{pmatrix} 0 & W_x & 0 & 0 \end{pmatrix}^{\mathrm T} \in L^{\omega \omega}$. Classes of these two elements generate $L^{\omega \omega} / L \cong \ZZ_n \times \ZZ_n$. One may check also that $L^\omega = L^{\omega \omega}$. Hence there exist $n^2$ superselection sectors containing a ground state and in each of these sectors the ground state is unique. This is different than for the toric code on a torus, for which there is only one superselection sector containing an $n^2$-dimensional space of ground states. This illustrates the difference in physical interpretations of modules $Z(\mathfrak C)$ and $S(\mathfrak C)$.

Let us also mention that in the present example $Q(\mathfrak C) \cong \ZZ_n \times \ZZ_n$, as on a~plane (but not on a torus). Even though the code is effectively one-dimensional (one direction being finite), this does not contradict Proposition \ref{prop:Q_dim} because $Z(\mathfrak C) \neq 0$.

\subsection{\texorpdfstring{$\ZZ_{p^t}$}{Zn} plaquette model}


Let $n = p^t$, where $p$ is a prime number and $t$ a positive integer. We consider a~$\ZZ_{p^t}$ version of Wen's plaquette model \cite{wen2003quantum} on a plane. Thus we take $P = R^2$ and let $L$ be the span of $s = \begin{pmatrix} 1 - xy & x - y \end{pmatrix}^{\mathrm T}$. $L$ is freely generated by $s$, so there exists an element $\varphi \in L^*$ such that $\varphi (s)=1$. Clearly $(x - y) \varphi$ and $(xy - 1) \varphi$ are representable by elements of $P$ and we have
\begin{equation}
    Q \cong R / (x - y, xy - 1).
    \label{eq:plaq_model_Q}
\end{equation}
There exists an abelian group isomorphism $Q \cong \ZZ_n \times \ZZ_n$ (as for the toric code), but in contrast to the case of toric code $Q$ is acted upon nontrivially by translations. Hence this model is in a different SET phase (with translational symmetry) than the toric code. On the other hand, these models are well-known to be equivalent if translational symmetry is ignored.

For a subgroup of $\Lambda$ acting trivially on $Q$, we can take the subgroup of index $4$ generated by $x^2, y^2$. With this choice, we found the following \v Cech cocycle $\varphi$ representing the generator of $Q$ (corresponding via the isomorphism \eqref{eq:plaq_model_Q} to the class of $1$):
\begin{equation}
    \varphi_1 = \frac{\begin{pmatrix} x & x^2 \end{pmatrix}^{\mathrm T}}{1-x^2}, \qquad \varphi_2 = \frac{\begin{pmatrix} -y & y^2 \end{pmatrix}^{\mathrm T}}{1-y^2}.
\end{equation}
Classes of cocycles $\varphi$ and $x \varphi$ form a $\ZZ_n$ basis of \v Cech cohomology.

We will find the decomposition of $Q$ into $\mathfrak m$-torsion parts. If $p \neq 2$, maximal ideals of $R$ containing the annihilator of $Q$ are of the form
\begin{equation}
    \mathfrak m^\pm = (p, x \mp 1, y \mp 1).
\end{equation}
The case $p=2$ is special because then $\mathfrak m^+ = \mathfrak m^-$. We assume that $p \neq 2$ from now on. $\mathfrak m^{\pm}$-torsion submodules of $Q$ correspond to cocycles $\varphi^\pm = (1 \pm x) \varphi$. They also form a $\ZZ_n$ basis of \v Cech cohomology. By Proposition \ref{prop:braiding_orthogonality}, $\varphi^+$ is $\Omega$-orthogonal to $\varphi^-$. Indeed, a calculation gives
\begin{equation}
    \Omega(\varphi, \varphi) = (x+y) \sum_{k,l \in \ZZ} x^{2k} y^{2l},
\end{equation}
and therefore
\begin{align}
    \Omega(\varphi^+ , \varphi^-) &= (1+ \overline x) (1-x) \Omega(\varphi, \varphi ) = 0, \\
    \Omega(\varphi^\pm , \varphi^\pm) &= (1 \pm \overline x) (1 \pm x) \Omega(\varphi, \varphi) = \pm 2 \sum_{k,l \in \ZZ} (\pm x)^k (\pm y)^l. \nonumber 
\end{align}

\begin{rmk}
Redefining $s$ to $\begin{pmatrix} 1 + xy & x+y \end{pmatrix}$ gives a second code, which is related to the one above by a local unitary transformation (which is $y^2$-invariant but not $y$-invariant). Simple calculation gives $Q \cong R / (x+y,xy+1)$, so this code is in a different SET phase than the previous one. 
\end{rmk}

\subsection{Haah's code and \texorpdfstring{$X$}{X}-cube model}

Haah's code and $X$-cube model (over $\ZZ_2$) are defined by
\begin{align}
\sigma_{\mathrm{Haah}} &= \begin{pmatrix}
1+xy+yz+zx & 0\\
1+x+y +z & 0\\
0 & 1+\bar x +\bar y +\bar z\\
0 & 1+\bar x \bar y + \bar y \bar z+\bar z \bar x
\end{pmatrix}, \\
\sigma_{X \text{-cube}} &= \begin{pmatrix}
1+\bar x +\bar y+\bar x \bar y & 0 & 0\\
1+ \bar y + \bar z + \bar y \bar z& 0 & 0\\
1+ \bar x + \bar z + \bar x \bar z& 0 & 0\\
0 & 1+z & 0\\
0 & 1+x & 1+x\\
0 & 0& 1+y
\end{pmatrix} \nonumber.
\end{align}
In both cases $\sigma$ is injective, so $L$ is free. This implies that $Q^i =0$ for $i>0$, so our approach confirms that corresponding phases of matter do not admit nontrivial spatially extended excitations. Computation of $Q^0$ of course agrees with what is known.

\subsection{\texorpdfstring{$\ZZ_2$}{Z4} toric code phase using \texorpdfstring{$\ZZ_4$}{Z4} coefficients}

We consider a code with composite-dimensional qudits which is nevertheless in the same phase as the $\mathbb Z_2$ toric code. Let $R = \ZZ_4[x^\pm,y^\pm]$, $P=R^4$ and $L = \im(\sigma)$, where
\begin{equation}
\sigma = 
\begin{pmatrix}
2+2\bar{x}& 0 & 0 & 0\\
2+2\bar{y}& 0 & 0 & 0\\
0 & 1-y & 2 & 0 \\
0 & 1-x & 0 & 2
\end{pmatrix}.
\end{equation}
Let us define a matrix
\begin{equation}
    \tau = 
    \begin{pmatrix} 
    2 & 0 & 0 & 0 \\
    0 & 2 & 0 & 0 \\
    0 & 1-y & 2 & 0 \\
    0 & 1-x & 0 & 2
    \end{pmatrix}.
\end{equation}
We have an infinite free resolution
\begin{equation}
    \dots \to R^4 \xrightarrow{\tau} R^4 \xrightarrow{\tau} R^4 \xrightarrow{\tau} R^4 \xrightarrow{\sigma} P \to P/L \to 0.
\end{equation}
Erasing $P/L$ and applying $(-)^*$ we obtain
\begin{equation}
0 \to P \xrightarrow{\sigma^\dagger \lambda}  R^4\xrightarrow{\tau^\dagger} R^4\xrightarrow{\tau^\dagger} R^4\rightarrow \cdots,
\end{equation}
from which one obtains
\begin{equation}
Q^0 \cong \ZZ_2 \oplus \ZZ_2, \qquad Q^i =0 \text{ for } i>0.
\end{equation}
$Q^0$ is generated by classes of $e = \begin{pmatrix} 2 & 0 & 0 & 0 \end{pmatrix}^{\mathrm{T}}$ and $m = \begin{pmatrix} 0 & 2 & 0 & 0 \end{pmatrix}^{\mathrm{T}}$, respectively. Following procedure outlined in 7.1, we find their corresponding cochain in $\Cech^1(\mathbf t, P)$: $$\varphi=\left( \begin{pmatrix} 0 \\ 0 \\ t_1^{-1} \\ 0 \end{pmatrix}, \begin{pmatrix}   0 \\ 0 \\ 0 \\ t_2^{-1}  \end{pmatrix} \right)$$ and $$\psi=\left( \begin{pmatrix} 0 \\ 2\overline{t}_1^{-1} \\ 0 \\ 0 \end{pmatrix}, \begin{pmatrix}   2\overline{t}_2^{-1} \\ 0 \\ 0 \\ 0  \end{pmatrix} \right).$$
Having these formulas in hand we evaluate the braiding
\begin{equation}
 \Omega(\varphi,\varphi)_0=\Omega(\psi,\psi)_0=0, \qquad   \Omega(\varphi, \psi )_0 = 2 \in \ZZ_4,
\end{equation}
and the topological spin
\begin{equation}
    \theta(\varphi)= \theta(\psi)=0, \qquad \theta(\varphi + \psi)=2.
\end{equation}

In spite of vanishing of higher $Q^i$, there exists no finite free resolution -- see characterization in Proposition \ref{prop:proj_Znfree} below. We remark that such phenomenon could only appear in models consisting of composite-dimensional qudits, and that its occurence is not an invariant of the topological phase. In fact such behavior is possible even for a model with a trivial (product state) ground state.

We also remark that $\Ext^i_R(L,R)=0$ for $i>0$. If $n$ was prime, we would be able to deduce from this that $L$ is a free module. In the present example, $L$ is not even quasi-free. Indeed, if $L$ was quasi-free, $L/2L$ would be a free module over $S=\ZZ_2[x^\pm,y^\pm]$. On the other hand, it is not difficult to check that $\Ext^1_S(L/2L,S) \neq 0$. This motivates the following result. 

\begin{prop} \label{prop:proj_Znfree}
Let $n=p^t$ for a prime number $p$. An $R$-module has finite projective dimension if and only if it is free over $\ZZ_n$. If this condition is satisfied, there exists a free resolution of length not exceeding $D$.
\end{prop}
\begin{proof}
$\implies :$ A projective $R$-module $P$ is a summand of a free $R$-module, which is clearly free over $\ZZ_n$. Thus $P$ is also projective over $\ZZ_n$. Projective modules over $\ZZ_n$ are free. 

Now let $P_\bullet \to M$ be a finite projective resolution of a module $M$. By the paragraph above, this is also a free resolution of $M$ considered as a $\ZZ_n$ module. Thus $M$ has finite projective dimension over $\ZZ_n$. Such $\ZZ_n$-modules are free.

$\impliedby :$ Let $M$ be free over $\ZZ_n$. We choose a $\ZZ_p[\Lambda]$-free resolution 
\begin{equation}
    0 \to P_D \to \dots \to P_1 \xrightarrow{\partial_1} P_0 \xrightarrow{\partial_0} M/pM \to 0
\end{equation}
of length $D$. This is possible by Hilbert's syzygy theorem. We will lift the resolution of $M/pM$ to a resolution of $M$ of the same length. Let $K_i = \ker(\partial_i)$.

For the purpose of this proof it will be convenient to denote reduction of an element mod $p$ by an overline. We have
\begin{equation}
    \partial_0 \begin{pmatrix} \overline r_1 \\ \vdots \\\ \overline r_n \end{pmatrix} = \sum_{i=1}^n \overline{r_i m_i} 
\end{equation}
for some $m_1, \dots, m_n \in M$ such that $\overline m_i$ generate $M/pM$. Then by Nakayama, $m_i$ generate $M$. Define $\widehat P_0 = R^n$ and $\widehat \partial_0 : R^n \to M$ by
\begin{equation}
    \widehat \partial_0 \begin{pmatrix}  r_1 \\ \vdots \\\ r_n \end{pmatrix} = \sum_{i=1}^n r_i m_i.
\end{equation}
By construction, $\widehat \partial_0$ is surjective. Let $\widehat K_0 = \ker(\widehat{\partial}_0)$. Reducing the short sequence $0 \to \widehat K_0 \to R^n \xrightarrow{\widehat \partial_0} M \to 0$ mod $p$ yields
\begin{gather}
    0 \to \widehat K_0 / p \widehat K_0 \to P_0 \xrightarrow{\partial_0} M / pM \to 0, \qquad \Tor_1^R(\widehat K_0, R/(p)) = 0.
    \label{eq:lifting_step}
\end{gather}
Here we used the simple fact that an $R$-module $N$ is free over $\ZZ_n$ if and only if $\Tor_1^R(N,R/(p))=0$, which can be verified using the resolution
\begin{equation}
    \dots \to R \xrightarrow{p} R \xrightarrow{p^{t-1}} R \xrightarrow{p} R \to R/(p) \to 0.
\end{equation}
Results in \eqref{eq:lifting_step} imply that $\widehat K_0 / p \widehat K_0$ may be identified with $K_0$ and $\widehat K_0$ is free over $\ZZ_n$.

Now replace $M$ by $K_0$ and $P_0$ by $P_1$ and repeat. Proceeding like this inductively we find short exact sequences
\begin{gather*}
    0 \to \widehat K_D \to \widehat P_D \to \widehat K_{D-1} \to 0 , \\
    \dots , \\
    0 \to \widehat K_0 \to \widehat P_0 \to M \to 0
\end{gather*}
such that each $P_i$ is free and $\widehat K_i / p \widehat K_i  \cong K_i$. In particular $\widehat K_D =0$ by Nakayama. Short sequences compose into a free resolution of $M$ of length $D$:
\begin{equation}
    0 \to \widehat P_D \to \dots \to \widehat P_0 \to M \to 0 .
\end{equation}
\end{proof}

\subsection{Double semion from condensing \texorpdfstring{$\ZZ_4$}{Z4} toric code}

We consider a code with composite-dimensional qudits. Let $R = \ZZ_4[x^\pm,y^\pm]$, $P=R^4$ and $L = \im(\sigma)$, where
\begin{equation}
\sigma = 
\begin{pmatrix}
    \bar{x}-1 & 0 & 0 & 2 \\
    \bar{y}-1 & 0 & 2x& 0\\
    1-y &  2+2y & 2 & 0 \\
    x-1 & 2+2x & 0 & 2\bar{y}
\end{pmatrix}.
\end{equation}
We define a matrix
\begin{equation}
    \tau = \begin{pmatrix}
        2 & 0 & 0 & 0 \\
        1 + \overline x \overline y & 2 & 0 & 0 \\
        \overline x  + \overline x \overline y & 0 & 2 & 0 \\ 
        1 + \overline x & 0 & 0 & 2
    \end{pmatrix}.
\end{equation}
We have an infinite free resolution
\begin{equation}
    \dots \to R^4 \xrightarrow{\tau} R^4 \xrightarrow{\tau} R^4 \xrightarrow{\tau} R^4 \xrightarrow{\sigma} P \to P/L \to 0.
\end{equation}
We apply $(-)^*$ and get
\begin{equation}
    0 \to P \xrightarrow{\sigma^\dagger \lambda} R^4 \xrightarrow{\tau} R^4 \xrightarrow{\tau} R^4 \to \dots,
\end{equation}
with matrices explicitly given by
\begin{align}
\sigma^\dagger \lambda &= \begin{pmatrix}
    1 - \overline y & \overline x -1 & 1 - x & 1 - y \\ 2 +  2 \overline y & 2 + 2 \overline x & 0 & 0 \\ 2 & 0 & 0 & 2 \overline x \\ 0 & 2 y & 2 & 0 
\end{pmatrix}, \\
\tau^\dagger &= \begin{pmatrix}
    2 & 1 + xy & x + xy & 1 + x \\ 0 & 2 & 0 & 0 \\ 0 & 0 & 2 & 0 \\ 0 & 0 & 0 & 2
\end{pmatrix}. \nonumber
\end{align}
It is easy to check that $\ker(\tau^\dagger) = \im(\tau^\dagger)$, so $\Ext^i(\overline{P/L}, R)=0$ for $i \geq 2$ and 
\begin{equation}
    \Ext^1 (\overline{P/L},R) \cong \im(\tau^\dagger) / \im(\sigma^\dagger \lambda).
    \label{eq:ds_ext}
\end{equation}
To compute this quotient, first note that the first two components of any vector in the image of $\sigma^\dagger \lambda$ are in the ideal $(x-1,y-1)$. Hence the first column, the second column, and the sum of the first two columns of $\tau^\dagger$ each represent nontrivial elements in the quotient. We claim that these are the three nonzero elements of $ \im(\tau^\dagger) / \im(\sigma^\dagger \lambda) \cong \mathbb Z_2 \oplus \mathbb Z_2$. Here is the proof:
\begin{itemize}
    \item The first column of $\tau^\dagger$ is annihilated by $2$, and the first column multiplied by $x-1$ (resp.\ $y-1$) is $2x$ times the second (resp. $2y$ times the first) column of $\sigma^\dagger \lambda$. Similarly one can verify that the class of the second column is annihilated by $2, x-1$, and $y-1$. This establishes that the images of the first two columns of $\tau^\dagger$ in the quotient \eqref{eq:ds_ext} generate a~submodule isomorphic to $\mathbb Z_2 \oplus \mathbb Z_2$.
    \item Next, we prove that the quotient is spanned by the elements described above. Note that $\begin{pmatrix}
        x (1 -y ) & 0 & 2 & 0 
    \end{pmatrix}^{\mathrm T}$ is the fourth column of $\sigma^\dagger \lambda$ multiplied by $x$. Modulo $xy$ times the first column of $\tau^\dagger$, this element is the third column of $\tau^\dagger$. Similarly, the fourth column of $\tau^\dagger$ is the sum of the third column of $\sigma^\dagger \lambda$ and the first column of $\tau^\dagger$ multiplied by $x$.
\end{itemize}

\v{C}ech cocycles $\varphi$ corresponding to the first two columns of $\tau^\dagger$ can be deduced from the reasoning above:
\begin{align}
    \varphi = \left(  \begin{pmatrix}
        0 \\ \frac{2x}{x-1} \\ 0 \\ 0 
    \end{pmatrix}, \begin{pmatrix}
        \frac{2y}{y-1} \\ 0 \\ 0 \\ 0
    \end{pmatrix} \right), \qquad \psi = \left( \begin{pmatrix}
        0 \\ -\frac{x}{x-1} \\ -\frac{xy}{x-1} \\ 0 
    \end{pmatrix}, \begin{pmatrix}
        \frac{y}{y-1} \\ 0 \\ 0 \\ -\frac{xy}{y-1}
    \end{pmatrix} \right).
\end{align}
In order to recognize the double semion phase, it is convenient to denote
\begin{equation}
    \psi_+ = \psi + \varphi, \qquad \psi_- = \psi.
\end{equation}
Using the explicit formulas for cocycles we find the braiding
\begin{equation}
    \Omega(\psi_\pm, \psi_\pm)_0 = 2 , \qquad \Omega(\psi_+,\psi_-)_0=0,
\end{equation}
and the topological spin
\begin{equation}
\theta(\psi_\pm) = \pm 1, \qquad \theta(\psi_++\psi_-) = 0.
\end{equation}

\appendix

\section{Gorenstein rings} \label{sec:Gorenstein}

\begin{defn} \label{def:Gorenstein}
Noetherian ring $A$ is called a Gorenstein ring if its injective dimension (as a module over itself) is finite. If it is zero, i.e. $A$ is an injective $A$-module, then $A$ is called a QF ring\footnote{QF stands for quasi-Frobenius.}.
\end{defn}

\begin{lem} \label{lem:Gor_loc_dim}
Let $A$ be a Gorenstein ring. Every localization of $A$ is a Gorenstein ring. The injective dimension of $A$ equals the Krull dimension of $A$. 
\end{lem}
\begin{proof}
Corollaries 1.3 and 5.6 in \cite{Bass}. 
\end{proof}

\begin{rmk}
It is popular to define the Gorenstein property first for Noetherian local rings and then declare a general Noetherian ring to be Gorenstein if its localization on any prime ideal is Gorenstein. Such rings do not necessarily have finite dimension. This situation is not encountered in this paper, so it is more convenient to stick to the more restrictive Definition \ref{def:Gorenstein}. 
\end{rmk}

\begin{lem} \label{lem:Gor_pol}
If a ring $A$ is Gorenstein, so is the polynomial ring $A[x]$. 
\end{lem}
\begin{proof}
Follows immediately from \cite[\href{https://stacks.math.columbia.edu/tag/0A6J}{Tag 0A6J}]{stacks-project}.
\end{proof}

\begin{lem} \label{lem:QF_mods}
Let $A$ be a QF ring. Every $A$-module $M$ embeds in a free module (of finite rank if $M$ is finitely generated). The natural module map $M \to \Hom_A(\Hom_A(M,A),A)$ is injective (an isomorphism if $M$ is finitely generated). In particular $M=0$ if and only if $\Hom_A(M,A) = 0$. 
\end{lem}
\begin{proof}
See \cite[Theorem 15.11]{Lam}.
\end{proof}

\begin{lem} \label{lem:dim_local}
Let $A$ be a commutative ring, $M$ an $A$-module and $r \geq 0$ an integer. If $\dim (M) \leq r$, then the localization $M_\mathfrak{p}$ vanishes for all prime ideals $\mathfrak p \subset A$ with $\dim (A/\mathfrak p) > r$. If $M$ is finitely generated, the converse is true. 
\end{lem}
\begin{proof}
Observe that $\dim (M) \leq r$ if and only if $\ann (M)$ is not contained in any prime ideal $\mathfrak p \subset A$ with $\dim (A/\mathfrak p) > r$. Suppose that this condition is satisfied and let $\mathfrak p$ be such that $\dim (R/\mathfrak p) > r$. Then $R \setminus \mathfrak p$ contains an element of $\ann (M)$, so~$M_\mathfrak p=0$. Next, let $M$ be finitely generated. Then $S^{-1} M =0$ for a multiplicative set $S \subset R$ if and only if $S \cap \ann (M) \neq \emptyset$. Thus $M_{\mathfrak p}=0$ for a~prime ideal $\mathfrak p$ if and only if $\ann (M)$ is not contained in $\mathfrak p$.
\end{proof}

\begin{lem} \label{lem:dim_Ext_bound}
Let $A$ be a Gorenstein ring of dimension $D$ and let $M$ be a~finitely generated $A$-module. Then $\dim (\Ext^i_A(M,A)) \leq D-i$.
\end{lem}
\begin{proof}
Let $\mathfrak p \subset A$ be a prime ideal with $\dim(A/ \mathfrak p) \geq D-i$. Then $A_{\mathfrak p}$ is a~Gorenstein ring with $A_{\mathfrak p} \leq D-i$, so $\Ext^{i+1}_A(M,A)_{\mathfrak p} \cong \Ext^{i+1}_{A_{\mathfrak p}}(M_{\mathfrak p},A_{\mathfrak p}) = 0$. Now invoke Lemma \ref{lem:dim_local}.
\end{proof}

\section{Local cohomology} \label{sec:loc_coh}

\begin{defn}
Let $A$ be a Noetherian commutative ring and $\mathfrak a \subset A$ an ideal. If $M$ is an $A$-module, $\Gamma_{\mathfrak a}(M) = \{ m \in M \, | \, \exists j \in \mathbb N \ \mathfrak a^j m=0 \}$ is called $\mathfrak a$-torsion submodule of $M$. Modules $M$ such that $M = \Gamma_{\mathfrak a}(M)$ are said to be $\mathfrak a$-torsion. $\Gamma_{\mathfrak a}$ is a~left exact functor. Its right derived functors $\H^j_{\mathfrak a}$ are called local cohomology functors. More explicitly, $\H^j_{\mathfrak a}(M)$ is defined as the $j$-th degree cohomology of the complex $\Gamma_a(I^\bullet)$, where $M \to I^{\bullet}$ is an injective resolution. 
\end{defn}

Note that by construction, every $\H^j_{\mathfrak a}(M)$ is a subquotient of an $\mathfrak a$-torsion module and hence is $\mathfrak a$-torsion. Moreover, $\H^0_{\mathfrak a}(M) \cong \Gamma_{\mathfrak a}(M)$.

\begin{lem} \label{lem:lc_rad_cop}
Let $M$ be an $A$-module.
\begin{enumerate}
    \item $\H^j_{\mathfrak a}(M) \cong \H^j_{\sqrt{\mathfrak a}}(M)$, where $\sqrt{\mathfrak a} = \{ a \in A \, | \, \exists j \in \mathbb N \ a^j \in \mathfrak a \}$.
    \item If $\mathfrak a_1, \dots ,\mathfrak a_t$ are coprime, then $\H^j_{\mathfrak {a}_1 \dots \mathfrak{a}_t}(M) \cong \bigoplus \limits_{i=1}^t \H^j_{\mathfrak a_i}(M) $.
\end{enumerate}
\end{lem}
\begin{proof}
1. As $A$ is Noetherian, $ (\sqrt{\mathfrak a})^N \subset \mathfrak a$ for some $N$, so $\Gamma_a = \Gamma_{\sqrt{a}}$.

2. By induction, for any $i \neq j$ and $k \in \mathbb N$ ideals $\mathfrak a^k_i, \mathfrak a^k_j$ are coprime. Letting $K(I) = \{ m \in M \, | \, I m=0 \}$ for an ideal $I$, Chinese remainder theorem gives $K(\mathfrak a_1^k \dots \mathfrak a_t^k) = \bigoplus_{i=1}^t K(\mathfrak a_i^k)$. Next use
$\Gamma_{\mathfrak a_1 \dots \mathfrak a_t}(M) = \bigcup_{k=0}^\infty K (\mathfrak a_1^k \dots \mathfrak a_t^k)$.
\end{proof}

\begin{lem} \label{lem:lc_codim_zero}
Let $\mathfrak a$ be an ideal such that $\dim(A/\mathfrak a)=0$. Then for any $A$-module $M$ and any $j$ we have a natural isomorphism
\begin{equation}
    \H^j_{\mathfrak a}(M) \cong \bigoplus_{\mathfrak m} \H^j_{\mathfrak m}(M),
\end{equation}
where the sum is over maximal ideals $\mathfrak m \subset A$ containing $\mathfrak a$.
\end{lem}
\begin{proof}
We have $\sqrt{\mathfrak a} = \bigcap_{\mathfrak m} \mathfrak m$. Moreover, there exists finitely many maximal ideals containing $\mathfrak a$ and they are pairwise coprime. In particular their intersection coincides with the product. We invoke Lemma \ref{lem:lc_rad_cop}. 
\end{proof}

\begin{lem} \label{lem:Gor_loc_ch}
Let $A$ be a local Gorenstein ring with maximal ideal $\mathfrak m$ and residue field $k$ and let $D$ be the dimension of $A$. Then $\H^j_{\mathfrak m}(A) =0$ for $j \neq D$ and $\H^D_{\mathfrak m}(A)$ is an injective envelope of $k$.
\end{lem}
\begin{proof}
See \cite[Theorem 11.26]{iyengartwenty}.
\end{proof}

\section{\v{C}ech complex} \label{sec:Cech}

Now let $\mathbf t = (t_1,\dots ,t_r)$ be a sequence of elements of $A$ and let $M$ be an $A$-module. We define in terms of its localizations
\begin{equation}
\Cech^0(\mathbf t, M) = M, \qquad \Cech^p(\mathbf t, M) = \bigoplus_{1 \leq i_1 < \dots < i_p  \leq r} M_{t_{i_1}\dots t_{i_p}}.
\end{equation}
If $\varphi \in \Cech^p(\mathbf t, M)$, we let $\varphi_{i_1 \dots i_p}$ be its component in $M_{t_{i_1} \dots t_{i_p}}$ for every sequence $1 \leq i_1 < \dots < i_p \leq r$. A differential $\delta : \Cech^p(\mathbf t , M) \to \Cech^{p+1}(\mathbf t , M)$ is defined by
\begin{equation}
    (\delta \varphi)_{i_0 \dots i_p} = \sum_{j=0}^p (-1)^j \varphi_{i_0 \dots i_{j-1} i_{j+1} \dots i_p},
\end{equation}
in which $\varphi_{i_0 \dots i_{j-1} i_{j+1} \dots i_p}$ is implicitly mapped from $M_{t_{i_0} \dots t_{i_{j-1}} t_{i_{j+1}} \dots t_{i_p}}$ to $M_{t_{i_0} \dots t_{i_p}}$ by the localization homomorphism. This makes $\Cech^\bullet(\mathbf t, M)$ a cochain complex. Its cohomology is denoted by $\CH^\bullet(\mathbf t, M)$ and called \v{C}ech cohomology.

\begin{lem} \label{lem:loc_ch_Cech}
One has $\CH^\bullet(\mathbf t,M) \cong \H^\bullet_{\mathfrak a}(M)$, where $\mathfrak a = (t_1, \dots, t_r)$.
\end{lem}
\begin{proof}
See \cite[Theorem 7.13]{iyengartwenty}.
\end{proof}

\begin{lem} \label{lem:loc_ch_localization}
Let $U$ be a multiplicatively closed subset of $A$, $A'= U^{-1} A$ and let $\mathfrak a'$ be the extension of $\mathfrak a$ in $A'$. Then for any $A$-module $M$ we have $\H^j_{\mathfrak a'}(U^{-1} M) \cong U^{-1} \H^j_{\mathfrak a}(M)$. 
\end{lem}
\begin{proof}
Follows from Lemma \ref{lem:loc_ch_Cech} because the corresponding property of \v{C}ech cohomology is easy to verify.
\end{proof}

Clearly we have $\Cech^\bullet(\mathbf t, M) = \Cech^{\bullet}(\mathbf t, A) \otimes_A M$. Since modules $\Cech^p(\mathbf t , A)$ are flat, this implies that $\Cech^\bullet(\mathbf t, -)$ takes short exact sequences of modules to short exact sequences of complexes. Hence every short exact sequence of modules induces a long exact sequence in \v{C}ech cohomology.

Let $\varphi \in \Cech^{p+1}(\mathbf t,M)$, $\psi \in \Cech^{q+1}(\mathbf t, N)$ with $p,q \geq 0$. We define the cup product $\varphi \smile \psi \in \Cech^{p+q+1}(\mathbf t, M \otimes_A N)$ by
\begin{equation}
(    \varphi \smile \psi )_{i_0 \dots i_{p+q}} = \varphi_{i_0 \dots i_p} \otimes \psi_{i_p \dots i_{p+q}}.
\end{equation}
It is associative and satisfies the graded Leibniz rule
\begin{equation}
    \delta(\varphi \smile \psi) = \delta \varphi \smile \psi + (-1)^p \varphi \smile \delta \psi,
\end{equation}
hence induces a product $\CH^{p+1}(\mathbf t, M) \otimes_A \CH^{q+1}(\mathbf t, N) \to \CH^{p+q+1}(\mathbf t, M \otimes_A N)$.

Let $\tau : N \otimes_A M \to M \otimes_A N$ be the standard isomorphism. For brevity we denote induced maps of \v{C}ech complexes and in \v{C}ech cohomology with the same symbol. Mimicking formulas in \cite{Steenrod} we define products
\begin{gather}
    \smile_1 \, : \, \Cech^{p+1}(\mathbf t, M) \otimes_A \Cech^{q+1}(\mathbf t,N) \to \Cech^{p+q}(\mathbf t, M \otimes_A N), \\
    (\varphi \smile_1 \psi)_{i_0 \dots i_{p+q-1}} = \sum_{j=0}^{p-1} (-1)^{(p-j)(q+1)} \varphi_{i_0 \dots i_j i_{j+q} \dots i_{p+q-1}} \otimes \psi_{i_j \dots i_{j+q}}  .\nonumber
\end{gather}
They satisfy the following identity:
\begin{align}
    &\varphi \smile \psi - (-1)^{pq} \tau (\psi \smile \varphi) \label{eq:cup_cup1_rel} \\
     = & (-1)^{p+q+1} \left[ \delta(\varphi \smile_1 \psi) - \delta \varphi \smile_1 \psi - (-1)^p \varphi \smile_1 \delta \psi \right]. \nonumber
\end{align}
If $\varphi, \psi$ are cocycles and $[\varphi], [\psi]$ are their cohomology classes, this gives
\begin{equation}
    [\varphi] \smile [\psi ] = (-1)^{pq} \tau \left( [\psi] \smile [\varphi] \right).
\end{equation}
In this sense the cup product is graded commutative. 

\begin{rmk}
The \v{C}ech complex and the cup product depend on the ordering of elements $t_i$. Howeover, cohomology (and the cup product in cohomology) do not. We refer for example to \cite[\href{https://stacks.math.columbia.edu/tag/01FG}{Tag 01FG}]{stacks-project} and discussion in \cite{Steenrod}.
\end{rmk}

\section*{Acknowledgments}
We would like to thank Anton Kapustin for his generous support and guidance. B.Y. is grateful to David Eisenbud and Hai Long Dao for discussions on Gorenstein and Cohen-Macaulay rings, and to David Cox, Anton Kapustin and Hal Schenck, whose advice led us to consider local cohomology. We thank Adam Artymowicz, Arpit Dua, Bailey Gu, Jeongwan Haah, Leszek Hadasz, Tamir Hemo, Xiuqi Ma, Sunghyuk Park and Nikita Sopenko for discussions. Work of B.R. on this project was initiated during his visit in California Institute of Technology. B.R. is grateful for hospitality. The visit was funded by the Kosciuszko Foundation. The American Centre of Polish Culture. Research of B.R. was also supported by the MNS donation for PhD students and young scientists N17/MNS/000040. B.Y. was supported in part by the Simons Foundation and the U.S. Department of Energy, Office of Science, Office of High Energy Physics, under Award Number DE-SC0011632. During final revisions of the paper B.R. was employed at the University of Copenhagen. \\
Data sharing not applicable to this article as no datasets were generated or analysed during the current study.

\printbibliography

@article{Fault,
  title={Fault-tolerant quantum computation by anyons},
  author={Kitaev, A Yu},
  journal={Annals of Physics},
  volume={303},
  number={1},
  pages={2--30},
  year={2003},
  publisher={Elsevier}
}

@article{McCoy,
  title={Remarks on divisors of zero},
  author={McCoy, N H},
  journal={The American Mathematical Monthly},
  volume={49},
  number={5},
  pages={286--295},
  year={1942},
  publisher={Taylor \& Francis}
}

@book{Lam,
  title={Lectures on modules and rings},
  author={Lam, Tsit-Yuen},
  volume={189},
  year={2012},
  publisher={Springer Science \& Business Media}
}

@book{Eisenbud,
  title={Commutative algebra: with a view toward algebraic geometry},
  author={Eisenbud, David},
  volume={150},
  year={2013},
  publisher={Springer Science \& Business Media}
}

@article{Bass,
  title={Injective dimension in Noetherian rings},
  author={Bass, Hyman},
  journal={Transactions of the American Mathematical Society},
  volume={102},
  number={1},
  pages={18--29},
  year={1962},
  publisher={JSTOR}
}

@misc{stacks-project,
    shorthand    = {Stacks},
    author       = {The {Stacks Project Authors}},
    title        = {\textit{Stacks Project}},
    howpublished = {\url{https://stacks.math.columbia.edu}},
    year         = {2018},
  }

@article{Steenrod,
  title={Products of cocycles and extensions of mappings},
  author={Steenrod, Norman E},
  journal={Annals of Mathematics},
  pages={290--320},
  year={1947},
  publisher={JSTOR}
}

@article{haah2011local,
  title={Local stabilizer codes in three dimensions without string logical operators},
  author={Haah, Jeongwan},
  journal={Physical Review A},
  volume={83},
  number={4},
  pages={042330},
  year={2011},
  publisher={APS}
}

@article{levin2003fermions,
  title={Fermions, strings, and gauge fields in lattice spin models},
  author={Levin, Michael and Wen, Xiao-Gang},
  journal={Physical Review B},
  volume={67},
  number={24},
  pages={245316},
  year={2003},
  publisher={APS}
}

@article{haah2013commuting,
  title={Commuting Pauli Hamiltonians as maps between free modules},
  author={Haah, Jeongwan},
  journal={Communications in Mathematical Physics},
  volume={324},
  number={2},
  pages={351--399},
  year={2013},
  publisher={Springer}
}

@article{haah2021classification,
  title={Classification of translation invariant topological Pauli stabilizer codes for prime dimensional qudits on two-dimensional lattices},
  author={Haah, Jeongwan},
  journal={Journal of Mathematical Physics},
  volume={62},
  number={1},
  pages={012201},
  year={2021},
  publisher={AIP Publishing LLC}
}

@article{Z_4,
  title = {Pauli Stabilizer Models of Twisted Quantum Doubles},
  author = {Ellison, Tyler D. and Chen, Yu-An and Dua, Arpit and Shirley, Wilbur and Tantivasadakarn, Nathanan and Williamson, Dominic J.},
  journal = {PRX Quantum},
  volume = {3},
  pages = {010353},
  year = {2022},
  publisher = {American Physical Society}
}

@book{iyengartwenty,
  title={Twenty-four Hours of Local Cohomology},
  author={Iyengar, S.},
  isbn={9780821872499},
  series={Graduate studies in mathematics},
  url={https://books.google.pl/books?id=5HgmUQsbe5sC},
  publisher={American Mathematical Soc.}
}

@article{bombin2014structure,
  title={Structure of 2D topological stabilizer codes},
  author={Bombin, Hector},
  journal={Communications in Mathematical Physics},
  volume={327},
  number={2},
  pages={387--432},
  year={2014},
  publisher={Springer}
}

@article{gaiotto2015generalized,
  title={Generalized global symmetries},
  author={Gaiotto, Davide and Kapustin, Anton and Seiberg, Nathan and Willett, Brian},
  journal={Journal of High Energy Physics},
  volume={2015},
  number={2},
  pages={1--62},
  year={2015},
  publisher={Springer}
}

@article{kapustin2011topological,
  title={Topological boundary conditions in abelian Chern--Simons theory},
  author={Kapustin, Anton and Saulina, Natalia},
  journal={Nuclear Physics B},
  volume={845},
  number={3},
  pages={393--435},
  year={2011},
  publisher={Elsevier}
}

@article{haah2021clifford,
  title={Clifford quantum cellular automata: Trivial group in 2D and Witt group in 3D},
  author={Haah, Jeongwan},
  journal={Journal of Mathematical Physics},
  volume={62},
  number={9},
  pages={092202},
  year={2021},
  publisher={AIP Publishing LLC}
}

@article{schlingemann2008structure,
  title={On the structure of Clifford quantum cellular automata},
  author={Schlingemann, Dirk-M and Vogts, Holger and Werner, Reinhard F},
  journal={Journal of Mathematical Physics},
  volume={49},
  number={11},
  pages={112104},
  year={2008},
  publisher={American Institute of Physics}
}

@article{wen2003quantum,
  title={Quantum orders in an exact soluble model},
  author={Wen, Xiao-Gang},
  journal={Physical review letters},
  volume={90},
  number={1},
  pages={016803},
  year={2003},
  publisher={APS}
}

@article{gottesman1996class,
  title={Class of quantum error-correcting codes saturating the quantum Hamming bound},
  author={Gottesman, Daniel},
  journal={Physical Review A},
  volume={54},
  number={3},
  pages={1862},
  year={1996},
  publisher={APS}
}

@article{calderbank1997quantum,
  title={Quantum error correction and orthogonal geometry},
  author={Calderbank, A Robert and Rains, Eric M and Shor, Peter W and Sloane, Neil JA},
  journal={Physical Review Letters},
  volume={78},
  number={3},
  pages={405},
  year={1997},
  publisher={APS}
}

@article{liang2023extracting,
  title={Extracting topological orders of generalized Pauli stabilizer codes in two dimensions},
  author={Liang, Zijian and Xu, Yijia and Iosue, Joseph T and Chen, Yu-An},
  journal={arXiv preprint arXiv:2312.11170},
  year={2023}
}

@article{ellison2023pauli,
  title={Pauli topological subsystem codes from Abelian anyon theories},
  author={Ellison, Tyler D and Chen, Yu-An and Dua, Arpit and Shirley, Wilbur and Tantivasadakarn, Nathanan and Williamson, Dominic J},
  journal={Quantum},
  volume={7},
  pages={1137},
  year={2023},
  publisher={Verein zur F{\"o}rderung des Open Access Publizierens in den Quantenwissenschaften}
}

@article{shirley2022three,
  title={Three-dimensional quantum cellular automata from chiral semion surface topological order and beyond},
  author={Shirley, Wilbur and Chen, Yu-An and Dua, Arpit and Ellison, Tyler D and Tantivasadakarn, Nathanan and Williamson, Dominic J},
  journal={PRX Quantum},
  volume={3},
  number={3},
  pages={030326},
  year={2022},
  publisher={APS}
}

\end{document}